\documentclass{article}

\usepackage[utf8]{inputenc}
\usepackage[T1]{fontenc}
\usepackage[utf8]{inputenc}
\usepackage[left=3cm, right=3.5cm]{geometry} 
\usepackage{hyperref}
\usepackage{enumitem}
\usepackage{amsmath, amssymb, amsthm}
\usepackage{thmtools}
\usepackage{mathtools}
\usepackage{stackengine}
\usepackage{tikz}
\usetikzlibrary{decorations.pathreplacing}
\usepackage{caption}
\usepackage{subcaption}
\theoremstyle{plain}
\newtheorem{thm}{\protect\theoremname}
\theoremstyle{plain}
\newtheorem{lem}[thm]{\protect\lemmaname}
\theoremstyle{plain}

\theoremstyle{plain}
\newtheorem{defn}[thm]{\protect\definitionname}

\providecommand{\definitionname}{Definition}
\providecommand{\corollaryname}{Corollary}
\providecommand{\lemmaname}{Lemma}
\providecommand{\theoremname}{Theorem}

\begin{document}
\global\long\def\N{\mathbb{N}}%
\global\long\def\R{\mathbb{R}}%
\global\long\def\calb{\mathcal{B}}%
\global\long\def\calp{\mathcal{P}}%
\global\long\def\calq{\mathcal{Q}}%
\global\long\def\calx{\mathcal{X}}%
\global\long\def\OPT{\mathtt{OPT}}%
\global\long\def\comp{\mathrm{comp}}%
\global\long\def\sr{\mathrm{sr}}%
\global\long\def\out{\mathrm{out}}%
\global\long\def\rad{\mathrm{rad}}%
\newcommand{\creplaced}[1]{\stackon[-2.7pt]{$\overline{#1}$}{\scriptsize$\circ$}}

\newcommand{\aw}[1]{\textcolor{orange}{#1}}
\newcommand{\hugo}[1]{\textcolor{red}{#1}}
\newcommand{\katja}[1]{\textcolor{blue}{#1}}
\newcommand{\moritz}[1]{\textcolor{green}{#1}}

\title{\Large A $(3+\epsilon)$-approximation algorithm for the minimum sum of radii problem with outliers and extensions for generalized lower bounds}
\author{Moritz Buchem\thanks{Technische Universität München.}
    \and Katja Ettmayr\footnotemark[1]
    \and Hugo K. K. Rosado\footnotemark[1]
    \and Andreas Wiese\footnotemark[1]}

\date{\today}

\maketitle

\begin{abstract}
    \small\baselineskip=9pt
    Clustering is a fundamental problem setting with applications in many different areas.
    For a given set of points in a metric space and an integer $k$, we seek to partition the given points into
    $k$ clusters. For each computed cluster, one typically defines one point as the center of the cluster.
    A natural objective is to minimize    the sum of the cluster center's radii, where we assign the smallest
    radius $r$ to each center such that each point in the cluster is at a distance of at most $r$ from the center.
    The best-known polynomial time approximation ratio for this problem is $3.389$.
    In the setting with outliers, i.e., we are given an integer $m$ and allow up to $m$ points that are not in any cluster, the best-known approximation factor~is~$12.365$.

    In this paper, we improve both approximation ratios to $3+\epsilon$.
    Our algorithms are primal-dual algorithms that use fundamentally new ideas to compute solutions and to guarantee the claimed approximation ratios.

    For example, we replace the classical binary search to find the best value of a Lagrangian multiplier~$\lambda$ by a primal-dual routine in which $\lambda$ is a variable that is raised.
    Also, we show that for each connected component due to almost tight dual constraints, we can find one single cluster that covers all its points and we bound its cost via a new primal-dual analysis.
    We remark that our approximation factor of $3+\epsilon$ is a natural limit for the known approaches in the literature.

    Then, we extend our results to the setting of lower bounds.
    There are algorithms known for the case that for each point $i$ there is a lower bound $L_{i}$, stating that we need to assign at least $L_{i}$ clients to $i$ if~$i$ is a cluster center.
    For this setting, there is a $ 3.83$ approximation if outliers are not allowed and a ${12.365}$-approximation with outliers.
    We improve both ratios to $3.5 + \epsilon$ and, at the same time, generalize the type of allowed lower~bounds.
\end{abstract}

\section{Introduction.}

Clustering is an important problem setting with applications in many areas, such as machine learning, image analysis, information retrieval, and data compression.
Given a set of data points, the goal is to group the points to clusters such that similar points (those that lie closely together) are assigned to the same cluster.
Typically, we are given an integer $k$ that denotes (an upper bound on) the desired number of clusters.
It is natural to define one point of each cluster to be its \emph{center}, and to define the \emph{radius} of the cluster to be the distance between the center and the point in the cluster that is furthest away from the center.

In order to assess the quality of a clustering, one needs a suitable objective function.
One possible objective function yields the $k$\textsc{-Center} problem in which we want to minimize the largest radius of the computed clusters.
This~problem is well understood: it admits a $2$-approximation algorithm and it is $\mathsf{NP}$-hard to approximate the problem within a factor of $2-\epsilon$, for any $\epsilon>0$~\cite{Gonzalez1985, Hochbaum1985}.
This~holds even in the setting of \emph{outliers} in which we are additionally given an integer $m$ and we do not need to include all points in the clusters, but we can omit up to $m$ points~\cite{Chakrabarty2020, Charikar2001}.
Therefore, the approximability of the $k$\textsc{-Center} problem is completely understood, even for the case of outliers.

However, in $k$\textsc{-Center}, the objective function value is dominated by the largest radius. In particular, if an instance requires one relatively large cluster, in an optimal solution all other clusters
might as well have the same (large) radius.
However, naturally one would like to choose the other radii as small as possible in such a case.

Therefore, in this paper, we study a different objective, which is to minimize the \emph{sum} of the cluster radii, which yields the \textsc{Minimum Sum of Radii} problem.
Formally, in this problem, we are given a set of points~$X$ in a metric space and an integer $k$.
For any two points $i,j\in X$ we denote by $d(i,j)$ their distance.
Our goal is to select a set of at most $k$ centers $i_{1},...,i_{\ell}$ and radii $r_{1},...,r_{\ell} \ge 0$ such that each point~$j\in X$ is \emph{covered}, i.e.,there is a center $i_{\ell'}$ with $\ell'\in[\ell]$ such that $d(j,i_{\ell'})\le r_{\ell'}$.
The objective function is to minimize~$\sum_{\ell'=1}^{\ell}r_{\ell'}$.
In the \textsc{Minimum Sum of Radii} \textsc{with Outliers} problem, we are given additionally an integer $m$ and allow up to $m$ outliers, i.e., we do not require that all points in $X$ are covered, but require
only that at least $|X|-m$ points in $X$ are covered.

The \textsc{Minimum Sum of Radii} problem is much less understood than the $k$\textsc{-Center} problem.
Without outliers, the best known polynomial time algorithm is due to Friggstad and Jamshidian and gives a $3.389$-approximation~\cite{Friggstad2022}, improving on a previous ${(3.504 + \epsilon)}$-approximation by Charikar and Panigrahy~\cite{Charikar2001a}.
Both algorithms are based on the classical concept of bi-point solutions (see, e.g.,~\cite{Ahmadian2016, Charikar2004, Friggstad2022, Jain2001, Li2013}).
The idea is to run a binary search on a Langrangian multiplier $\lambda$ and to compute two solutions, one with at least $k$ centers (which might have strictly more than $k$ centers, which makes it infeasible) and one with at most $k$ centers.
Then, one ``merges'' these two solutions into one feasible solution (with at most $k$ centers); intuitively, one tries to compute a convex combination of them. For these known approaches, the natural limit is an approximation factor of 3, which one would obtain if one could compute a ``perfect'' convex combination of the two candidate
solutions. However, the respective approximation ratios are worse than $3$ since it is not clear how to compute such a solution; in fact, it might not even exist.

In the setting with outliers, \textsc{Minimum Sum of Radii} is even less understood.
The best-known result is a $12.365$-approximation algorithm by Ahmadian and Swamy~\cite{Ahmadian2016}.
We note that this algorithm works even for a more general problem in which for each point $i\in X$ there is additionally a lower bound $L_{i}$, specifying that if $i$ is the center of a cluster, then at least $L_{i}$ points must be assigned to $i$. In particular, in addition to choosing the cluster centers and their radii, for each point $j\in X$ we need to choose a cluster center $i$ with~$d(i,j)\le r_{i}$ that we assign $j$ to (obeying the required lower bounds).
In the setting with lower bounds but without outliers, there is a $3.83$-approximation due to Ahmadian and Swamy~\cite{Ahmadian2016}.

A related problem is the \textsc{Minimum Sum of Diameters} problem in which we want to partition the given points into $k$ clusters for a given integer $k$, and we seek to minimize the sum of the diameters of the clusters, rather than the radii.
Any $\alpha$-approximation algorithm for \textsc{Minimum Sum of Radii} immediately yields a $2\alpha$-approximation
for \textsc{Minimum Sum of Diameters} problem, and thus we obtain algorithms for this problem via the results for \textsc{Minimum Sum of Radii} above (for the respective cases).
For the setting without outliers and without lower bounds, a $6.546$-approximation is known due to Friggstad and Jamshidian~\cite{Friggstad2022}, which is hence better than the approximation ratio obtained via the best-known algorithm for \textsc{Minimum Sum of Radii} for this~case.

\subsection{Our contribution.}

In this paper, we present a polynomial time ${(3+\epsilon)}$-approximation for the \textsc{Minimum Sum of Radii} problem with \emph{and} without outliers.
Thus, we essentially reach the mentioned limit of $3$, even for the case with outliers for which the best known previous approximation ratio was ${12.365}$~\cite{Ahmadian2016}.
We also study the setting with lower bounds for which we obtain a ${(3.5+\epsilon)}$-approximation algorithm with and without outliers, again improving the best known ${12.365}$-approximation and ${3.83}$-approximation algorithms~\cite{Ahmadian2016}.
In fact, our algorithm works even in a more general setting which we call \emph{generalized lower bounds}.
See~Table~\ref{tab:results} for a list of our results.

\begin{table}
    \setlength{\tabcolsep}{4pt} \resizebox{1.\textwidth}{!}{
        \begin{tabular}{|c|cc|cc|}
            \hline
            \textsc{Min. Sum of Radii}             & \multicolumn{2}{c|}{no lower bounds} & \multicolumn{2}{c|}{with lower bounds}\tabularnewline
            \hline
            \multicolumn{1}{|c|}{without outliers} & ~\textbf{\boldmath$3+\epsilon$}      & ($3.389$~\cite{Friggstad2022})                        & ~\textbf{\boldmath$3.5+\epsilon$} & ($3.83$~\cite{Ahmadian2016}) \tabularnewline
            \hline
            \multicolumn{1}{|c|}{with outliers}    & ~\textbf{\boldmath$3+\epsilon$}      & ($12.365$~\cite{Ahmadian2016})                        & ~\textbf{\boldmath$3.5+\epsilon$} & ($12.365$~\cite{Ahmadian2016}) \tabularnewline
            \hline
        \end{tabular}\hspace{11pt} %
        \begin{tabular}{|ccc|}
            \hline
            \multicolumn{3}{|c|}{\textsc{Min. Sum of Diameters}}\tabularnewline
            \hline
            \multicolumn{1}{|c|}{without outliers} & ~\textbf{\boldmath$6+\epsilon$} & ($6.546$~\cite{Friggstad2022})\tabularnewline
            \hline
            \multicolumn{1}{|c|}{with outliers}    & ~\textbf{\boldmath$6+\epsilon$} & ($24.73$~\cite{Ahmadian2016})\tabularnewline
            \hline
        \end{tabular}}\caption{\label{tab:results}Our approximation ratios compared to the previously best-known results (in parenthesis).}
\end{table}

First, we present our $(3+\epsilon)$-approximation for \textsc{Minimum Sum of Radii} without outliers.
It differs drastically from the mentioned previous approaches for the problem~\cite{Charikar2001a,Friggstad2022}.
In contrast to those, we do \emph{not} do a binary search on the Lagrangian multiplier $\lambda$.
Instead, we use a primal-dual approach in which $\lambda$ is a \emph{variable}, rather than a constant, and $\lambda$
is raised like all the other variables in the dual linear program~(LP).
In the process, there are pairs $(i,r)$ for a point $i\in X$ and a radius $r$ whose corresponding dual constraints become (essentially) tight; we select their corresponding variables in the primal LP.
For such a pair $(i,r)$, one may visualize a ball of radius $r$ around $i$.
At the end, we obtain one single value for $\lambda$ and a solution for the primal LP such that the balls of the selected tight pairs $(i,r)$ form at most $k$ connected components (in the graph obtained by introducing a vertex for each center and connecting two centers by an edge if their corresponding balls overlap).
However, there is one special selected pair $(i^{*},r^{*})$ such that without $(i^{*},r^{*})$, we would obtain strictly more than $k$ components.
Thus, intuitively, we obtain two solutions that are very similar: one solution including $(i^{*},r^{*})$ with at most $k$ components, and another solution without $(i^{*},r^{*})$ with more than $k$ components.

In the previous approaches, two candidate solutions are computed (that are possibly much less similar than ours).
Then, they are transformed into integral solutions as follows.
The selected pairs in the primal LP are ordered non-increasingly by their respective radii $r$.
Then, for each considered pair $(i,r)$, one creates a cluster with center $i$ and radius $3r$ (this is why an approximation factor of 3 is a natural limit), and one deletes all pairs $(i',r')$ for which the ball corresponding
to~$(i',r')$ intersects with the ball corresponding to~$(i,r)$ (i.e., for which there is a point $j\in X$ that is at distance at most $r$ to $i$ and at distance at most $r'$ to~$i'$).
However, even if one starts with two solutions that differ by only one selected pair like $(i^{*},r^{*})$ in the primal LP, the resulting solutions can look very differently due to cascading effects.
Then, it is not clear how to merge them while keeping the resulting cost bounded.

Instead, we use a different and algorithmically much simpler approach.
For each connected component of selected tight pairs $(i,r)$, we define \emph{one} cluster that contains all points covered by pairs in the component.
One key technical contribution is to show that we can choose a corresponding cluster center and a radius such that the resulting cost is paid by the dual variables corresponding to the points in the cluster.
This allows for a very simple algorithmic routine to select the cluster center and its radius: we simply choose the smallest such pair that covers all points of the component.
Recall that without the special pair $(i^{*},r^{*})$ we would have more than $k$ components.
We use this property to argue via a primal-dual analysis that our approximation ratio is $3+\epsilon$. In particular, we need at least $k$ components in order to pay for a term in dual objective corresponding to the Lagrangian multiplier $\lambda$.

Then we present our $(3+\epsilon)$-approximation algorithm for the setting with outliers.
Again, we compute one value for the Lagrangian multiplier $\lambda$ and a solution to the primal~LP with at most
$k$ components.
However, the setting with outliers is more complicated, and because of this, we need a solution with more technical properties to make the primal-dual analysis work.
Before, we had only one single special pair~$(i^{*},r^{*})$ that we included or not, yielding two different solutions.
Now, we order the selected tight pairs and construct a new solution by taking a well-chosen prefix of them, possibly adding a special pair~$(i^{*},r^{*})$, similarly as~before.
Our algorithm for computing the primal solution is different from the case without outliers.
For each fixed value for the Lagrangian multiplier $\lambda$, we consider a primal-dual process in which the dual variables are raised simultaneously and we stop raising a dual variable once one of its constraints becomes tight (similar to~\cite{Ahmadian2016,Charikar2001a}).
It depends on the initial (fixed) value of $\lambda$ how this process behaves, which dual constraints become tight, in which order they become tight, etc.
We do a binary search on $\lambda$; however, we do it over discrete break-points at which intuitively the mentioned process starts to behave differently, e.g., because the order changes in which the dual constraints become tight.
Once the binary search has found the ``right'' value for $\lambda$, we need to perform additional steps to ensure the mentioned technical properties for our solution.

In the end, we obtain one value for $\lambda$ and a set of selected pairs $(i,r)$ in the primal LP.
Our algorithm is then again very easy: we consider the connected components formed by these pairs and form one cluster for each of them.
However, the difficult part is the analysis where we need the additional properties of the computed solution to bound our approximation ratio.
In particular, the dual objective contains a term corresponding to $\lambda$ and an additional term corresponding to the outliers.
Intuitively, both terms need to be paid for in the primal-dual analysis, so that the cost of our primal solution is essentially by at most a factor~$3+\epsilon$ larger than the constructed dual solution.

As mentioned above, the ratio of $3+\epsilon$ is the natural limit for the previous approaches to solve \textsc{Minimum Sum of Radii}, even in the setting without outliers~\cite{Charikar2001a,Friggstad2022}.
In particular, we provide a family of tight instances for the class of primal-dual algorithms that are based on
raising the dual variables simultaneously (while using a binary search framework to estimate $\lambda$).
This also yields a tight example for our~approach (see Appendix~\ref{apx:tight_example}).

Finally, we extend our algorithms from above to $(3.5+\epsilon)$-approximation algorithms (with and without outliers) for a setting that we call \emph{generalized lower bounds}.
It comprises many cases that go beyond requiring for each center $i$ that at least $L_{i}$ clients are assigned to $i$, for some given value $L_{i}$.
For example, we can handle weighted clients with different colors such that each possible center has a lower bound for the needed weight of the assigned clients from each color class.
Also, it captures a generalization of \textsc{Minimum Sum of Radii} in which only a subset of the given points are allowed as cluster centers.
Formally, we can handle any setting in which for each possible center $i\in X$ there is a (maybe implicitly given) set~${\calx_{i}\subseteq2^{X}}$ of allowed sets of clients for~$i$, such that if $X'\in\calx_{i}$ then any superset of $X'$ is also contained in~$\calx_{i}$.
Our techniques above are sufficiently robust such that we can extend them to this new setting.
The main difference is that in the primal LP we include a variable for a pair $(i,r)$ only if there are enough points at distance $r$ from $i$ such that the lower bound requirement can be satisfied; also, we select the cluster centers
for our solution only from centers of pairs in the support of the LP.
This results in the (only slightly larger) approximation ratio of $3.5+\epsilon$.

Our results above for \textsc{Minimum Sum of Radii} immediately yield $(6+\epsilon)$-approximation algorithms for
\textsc{Minimum Sum of Diameters} with and without outliers.
This improves the best-known approximation ratios for the problem in both settings.
Note that the setting of lower bounds is not well-defined for \textsc{Minimum Sum of Diameters} since there are no cluster centers.

\subsection{Other related work.}

For \textsc{Minimum Sum of Radii} (without outliers and without lower bounds) there is a $(3+\epsilon)$-pseudo-approximation algorithm due to Charikar and Panigrahy that computes a pseudo-solution with up to
$k(1+1/\epsilon)$ clusters, but compares its cost to an optimal solution with only $k$ clusters~\cite{Charikar2001a}.
For the special case of a shortest-path metric in an unweighted graph, the problem is even solvable in polynomial time in the special setting where radii of length $0$ are not allowed~\cite{Behsaz2015}, while for the case of a shortest-path metric in an edge-weighted graph, the problem admits an $\mathsf{XP}$ algorithm with parameter treewidth~\cite{Drexler2023}.
Also, the problem can be solved in polynomial time in Euclidean spaces of constant dimension~\cite{Gibson2010}, but it is $\mathsf{NP}$-hard in shortest-path metrics induced by weighted planar graphs, in metrics of doubling dimension, and even when the locations of the centers are known beforehand~\cite{Drexler2023, Gibson2010}.
Moreover, the problem admits a randomized $(1+\epsilon)$-approximation algorithm in quasi-polynomial time~\cite{Gibson2010}.
This directly implies in a randomized $(2+\epsilon)$-approximation algorithm in quasi-polynomial time for the \textsc{Minimum Sum of Diameters}, which is $\mathsf{NP}$-hard to approximate within a factor of $2-\epsilon$, for any $\epsilon>0$~\cite{Doddi2000}.
For the capacitated version (i.e., each cluster center can serve only a bounded number of clients) there are FPT approximation algorithms known with constant approximation ratios~\cite{Bandyapadhyay2023,Inamdar2020},
while no $O(1)$-approximation algorithm is known with polynomial running~time.

There are also other objective functions for clustering studied in the literature yielding, e.g., the $k$-\textsc{Median} problem (where we minimize the sum of the distances of each point to its respective nearest center) and the $k$-\textsc{Means} problem (where we minimize the sum of the \emph{squared} distances of each point to its nearest center).
For $k$-\textsc{Median}, Byrka et al. provided the currently best approximation ratio of $2.675+ \epsilon$~\cite{Byrka2017}, while for $k$-\textsc{Means} the best-known approximation ratio is $6.357$ due to Ahmadian et al. \cite{Ahmadian2020}.

\section{$(3+\epsilon)$-approximation algorithm without outliers.}
\label{sec:3-approx}

In this section, we present our $(3+\epsilon)$-approximation algorithm for the \textsc{Minimum Sum of Radii} problem, i.e., for the setting without outliers.

Let $\epsilon>0$ and assume w.l.o.g.~that $1/\epsilon$ is an integer.
Assume we are given an instance of \textsc{Minimum Sum of Radii} and let $\OPT$ denote its optimal solution.
Note that w.l.o.g.~we can assume that in $\OPT$, each radius $r_{\ell}$ of some center $i_{\ell}$ equals $d(i_{\ell},j)$ for some point $j\in X$ (since otherwise we could reduce~$r_{\ell}$).
Hence, if $k\le1/\epsilon$ we can solve the problem by complete enumeration in time $n^{O(1/\epsilon)}$.
Assume now that~$k>1/\epsilon$. We guess the $1/\epsilon$ centers $i_{1},...,i_{1/\epsilon}$ from $\OPT$ with largest radii and their corresponding radii~$r_{1},...,r_{1/\epsilon}$, where $r_{1} \geq \dots \geq r_{1/\epsilon}$.
Thus, the remaining instance consists of the points $X':=X\setminus\bigcup_{\ell=1}^{1/\epsilon}B(i_{\ell},r_{\ell})$
where for each pair of a point $i\in X$ and a value $r\ge0$ we define $B(i,r):=\{j\in X:d(i,j)\le r\}$.
Also, in the remaining instance, we can open up to $k':=k-1/\epsilon$ centers, and, w.l.o.g., we assume that $|X'| > k'$.
Although a solution for instance $(X',k')$ needs to cover only $X'$, we allow solutions to have centers in $X$.

\begin{lem}
    \label{lem:biggest-radius}
    For the instance $(X',k')$ there exists an optimal solution with some centers $i_{1},...,i_{\ell}\in X$ and radii $r_{1},...,r_{\ell}$ such that $r_{\ell'}\le r_{1/\epsilon}\le\epsilon\cdot\OPT$ for each $\ell'\in[\ell]$.
\end{lem}
\begin{proof}
    For the instance $(X', k')$, let $\OPT'$ be an optimal solution.
    Let $\calb'$ denote the solution for $(X',k')$ that is obtained from $\OPT$ by disregarding the $1/\epsilon$ centers with largest radii and their corresponding radii.
    Assume for contradiction that $\calb'$ is not an optimal solution for $(X',k')$; then, $\OPT'$ together with these $1/\epsilon$ centers and corresponding radii form a feasible solution to $(X,k)$ of strictly smaller cost than $\OPT$, giving a contradiction.
    Thus, $\calb'$ has to be an optimal solution to $(X',k')$.
    Now, because the largest radius $r'$ of a center in $\calb'$ is smaller than the largest $r_1,\dots,r_{1/\epsilon}$ largest radii of centers in $\OPT$, it follows that
    \[
        \OPT
        \geq \sum_{i=1}^{1/\epsilon} r_i
        \geq \frac{1}{\epsilon} r_{1/\epsilon}
        \geq \frac{1}{\epsilon} r',
    \]
    completing the proof.
\end{proof}

In the following, we solve the instance $(X',k')$.
Let $\OPT'$ denote its optimal solution.
Since we may disregard radii that are larger than $r_{1/\epsilon}$, we define $\calb:=\{(i,r):i\in X \wedge \exists j\in X' \text{ s.t. } r=d(i,j) \leq r_{1/\epsilon} \}$ which hence contains all pairs $(i,r)$ in $\OPT'$.
Observe that ${|\calb| \in O(|X|^2)}$.
We say that a pair $(i,r)$ \emph{covers} a point $j$ if $j \in B(i,r)$ and that a set of pairs $\calb'$ \emph{covers} $j$ if $j \in B(i,r)$ for some $(i,r) \in \calb'$.
We consider the LP-relaxation $(P)$ of the remaining problem where each variable $x_{i,r}$ indicates whether the pair $(i,r)$ is chosen, i.e., whether the point $i$ is selected as a center with radius $r$.
We denote by $(D)$ its dual and define~$\OPT'_{LP}$ to be the optimal solution value of both~LPs.
Observe that $\OPT'_{LP} \leq \OPT' \leq \OPT$.

\vspace*{-14pt}
\begin{minipage}[t]{1\textwidth}
    \begin{minipage}[t]{0.45\textwidth}
        \begin{alignat*}{3}
            (P)\quad & \text{min}  & \displaystyle \sum_{(i,r)\in\calb}r\cdot x_{i,r} \hspace*{-13pt}         &         & \quad &                         \\
                     & \text{s.t.} & \displaystyle \sum_{(i,r) \in \calb:j\in B(i,j)} \hspace*{-17pt} x_{i,r} & \geq 1  &       & \forall j \in X'        \\
                     &             & \displaystyle \sum_{(i,r)\in\calb} x_{i,r}                               & \leq k' &       &                         \\
                     &             & x_{i,r}                                                                  & \geq 0  &       & \forall (i,r) \in \calb
        \end{alignat*}
    \end{minipage}
    \quad
    \begin{minipage}[t]{0.45\textwidth}
        \begin{alignat*}{3}
            (D)\quad & \text{max}  & \displaystyle \sum_{j\in X'}\alpha_{j}-k'\cdot\lambda \hspace*{-32pt} &                  & \quad &                        \\
                     & \text{s.t.} & \displaystyle \sum_{j\in B(i,r) \cap X'} \hspace*{-10pt} \alpha_{j}   & \leq r + \lambda &       & \forall(i,r) \in \calb \\
                     &             & \alpha_{j}                                                            & \geq 0           &       & \forall j \in X'       \\
                     &             & \lambda                                                               & \geq 0           &       &
        \end{alignat*}
    \end{minipage}
\end{minipage}
\medskip

Our goal is now to compute a solution $\left(\lambda,\alpha\right)$ to $(D)$ with a certain structure.
For a given solution $\left(\lambda,\alpha\right)$, we define that for a pair $(i,r)\in\calb$ the corresponding dual constraint is almost tight if $\sum_{j\in B(i,r)\cap X'}\alpha_{j}\ge r+\lambda-\mu$ for $\mu:=\frac{r_{1/\epsilon}}{|X|^2}$.
In this case, we say that this pair $(i,r)$ is \emph{almost tight}\footnote{
    At first sight, it might seem more intuitive to talk about almost tight \emph{balls}, rather than almost tight pairs.
    However, this could lead to ambiguities since balls are point sets: it is possible that for two pairs $(i,r),(i'r')$ with $r>r'$ we have $B(i,r)=B(i',r')$ while $(i,r)$ is almost tight but $(i',r')$ is not.}.
Intuitively, we will be interested in the \emph{components} formed by these almost tight pairs.
These correspond to the connected components of the graph we obtain by introducing one vertex for each tight pair,
and connecting two vertices by an edge if the corresponding balls share a point.

\begin{defn}
    \label{defn:components}
    Let $\calb'\subseteq\calb$ denote a set of pairs.
    Define a graph $G$ with $V(G):=\{v_{i,r}:(i,r)\in\calb'\}$ and $E(G):=\left\{ \{v_{i,r},v_{i',r'}\}:B(i,r)\cap B(i',r')\ne\emptyset\right\} $.
    A set of pairs $(i_{1},r_{1}),...,(i_{\ell},r_{\ell})\in\calb'$ forms a \emph{component} of $\calb'$ if the vertices $v_{i_{1},r_{1}},...,v_{i_{\ell},r_{\ell}}$ form a connected component of $G$, and $\comp(\calb')$ denotes the collection of components~of~$\calb'$.
\end{defn}

In the following, we are interested in solutions to $(D)$ that have a certain structure.
This structure will help us later to transform such solutions into feasible solutions to our given instance, while increasing the cost essentially only by a factor of at most $3+\epsilon$.

\begin{defn}
    Let $\left(\lambda,\alpha\right)$ be a solution to $(D)$ and let $\calb'\subseteq\calb$.
    We say that \emph{$\calb'$} is a set of \emph{structured pairs} for~$\left(\lambda,\alpha\right)$ if there is an almost tight pair $(i^{*},r^{*}) \in \calb'$ such that
    \begin{enumerate}[label=\normalfont SP$\arabic*.$,ref={\mbox{\rm SP$\arabic*$}}, noitemsep, leftmargin=1.5cm]
        \item each $(i,r)\in\calb'$ is almost tight in $\left(\lambda,\alpha\right)$,
        \item $\calb' \setminus \{ (i^*,r^*) \}$ covers $X'$,
        \item $|\comp(\calb'\setminus\{(i^{*},r^{*})\})| > k' \geq |\comp(\calb')|$.
    \end{enumerate}
\end{defn}

We will show that we can compute solutions for which there is a corresponding set of structured pairs~$\calb'$.
In particular, we will prove the following lemma in Subsection~\ref{subsec:Computing-well-structured}.

\begin{restatable}{lem}{lemcomputenooutliers}
    \label{lem:compute}
    In polynomial time, we can compute a solution $\left(\lambda,\alpha\right)$ to $(D)$, together with a set $\calb'$ of structured pairs for $\left(\lambda,\alpha\right)$.
\end{restatable}

Moreover, we will show that given a solution $\left(\lambda,\alpha\right)$ to $(D)$ and a corresponding set of structured pairs, we can compute a feasible solution whose cost is essentially by at most a factor $3+\epsilon$ larger than the objective function value of $\left(\lambda,\alpha\right)$.
We will prove the following lemma in Subsection~\ref{subsec:Rounding-well-structured}.

\begin{restatable}{lem}{lemcomputenooutliersfinalcost}
    \label{lem:no-outliers-final-cost-bound}
    Let $\calb'$ be a set of structured pairs for a solution $\left(\lambda,\alpha\right)$ to $(D)$.
    In polynomial time we can compute a feasible solution whose cost is at most $3 \cdot \OPT'_{LP} + O(\epsilon) \cdot \OPT \leq 3\cdot \OPT' + O(\epsilon) \cdot \OPT$.
\end{restatable}

The solution due to Lemma~\ref{lem:no-outliers-final-cost-bound} and the guessed pairs $(i_{1},r_{1}),...,(i_{1/\epsilon},r_{1/\epsilon})$ from $\OPT$ together yield a solution to the original instance $(X,k)$ of cost at most $(3+O(\epsilon))\OPT$.

\begin{thm}
    \label{thm:3+eps}
    For any $\epsilon>0$ there is a $(3+\epsilon)$-approximation algorithm for the \textsc{Minimum Sum of Radii}~problem.
\end{thm}

We show in Appendix~\ref{apx:tight_example} that this approximation factor is tight.

\subsection{\label{subsec:Computing-well-structured} Computing a solution and well-structured pairs.}

We describe an algorithm computing a solution to $(D)$ and a corresponding set of structured pairs, yielding the proof of Lemma~\ref{lem:compute}.

We start with a solution to $(D)$ in which all variables are zero, i.e., $\alpha_{j}=0$ for each point $j\in X'$ and also~$\lambda=0$.
Observe that for each point $j\in X'$ the pair $(j,0)$ is almost tight and that $j\in B(j,0)$.
If the almost tight pairs form at most $k'$ components, then we are essentially done. We initialize $\calb':= \{(j,0) : j\in X'\}$ (having thus strictly more than $k'$ components).
Then, we greedily add more almost tight pairs to~$\calb'$ until the number of connected components is at most $k'$.

Suppose now that the almost tight pairs form more than $k'$ connected components and that every point $j \in X'$ is covered by at least one almost tight pair.
Our algorithm runs in iterations.
In each iteration, we first select a maximal set $X''\subseteq X'$ such that each almost tight pair $(i,r)$ covers at most one point in $X''$.
We do this greedily: we initialize $X'':=\emptyset$ and consider the points in~$X'$ in an arbitrary order. We insert a point $j\in X'$ into $X''$ if all almost tight pairs $(i,r) \in \calb$ covering $j$ do not cover some other point that is already in $X''$.
Let $X''$ denote the resulting set. Then, we uniformly raise~$\lambda$ and $\alpha_{j}$ for each $j\in X''$ until
there is an almost tight pair $(i,r)$ that contains at least two points from~$X''$.
Formally, we define a value $\delta >0$ as the maximum amount by which we can raise the mentioned variables without violating a dual constraint, which is
\[
    \delta=\min_{(i,r) \in \calb} \left\{ \frac{r + \lambda - \sum_{j \in B(i,r)\cap X'} \alpha_j}{|B(i,r) \cap X''|-1} : |B(i,r) \cap X''|\geq 2 \wedge (i,r) \text{ is not almost tight} \right\}.
\]

We update $\lambda:=\lambda+\delta$ and for each $j\in X''$ we update $\alpha_{j}:=\alpha_{j}+\delta$.
For each $j\in X'\setminus X''$, we do not change the value of $\alpha_{j}$.
After each iteration, we consider the components of the almost tight pairs of the current solution.
We stop if they form at most $k'$ components.
Furthermore, by the choice of $X''$ every point~$j\in X'$ is covered by at least one almost tight pair at the end of the iteration since it is covered by at least one almost tight pair at the start of an iteration, thus the set of almost tight pairs always cover $X'$.
Note that each almost tight pair contains at most one point from $X''$; hence, for no such pair the corresponding dual constraint can be violated for any choice of $\delta>0$.
Furthermore, because of the slack in the constraints of the almost tight pairs, one can show that in each iteration the objective function increases by at least~$\frac{\mu}{|X|} = \frac{r_{1/\epsilon}}{|X|^3}$, and since ${\OPT'_{LP} \leq k' \cdot r_{1/\epsilon}}$, the number of iterations bounded by~$O(|X|^4)$.

It remains to describe how we define the set of structured pairs $\calb'$.
Let $\calb_{0}$ be the set of almost tight pairs after the penultimate iteration; thus, they form strictly more than $k'$ components.
Let $\calb_{1}$ be the set of almost tight pairs after the last iteration; they form at most $k'$ components.
We initialize $\calb':=\calb_{0}\cap\calb_{1}$.
We consider the almost tight pairs in $\calb_{1}\setminus\calb_{0}$ in an arbitrary order.
We add these pairs to $\calb'$ until the number of connected components is at most $k'$.
Let $(i^{*},r^{*})$ denote the last added pair.
Thus, $\calb'$ has at most $k'$ connected components, while $\calb'\setminus\{(i^{*},r^{*})\}$ has more than $k'$~components.

\begin{proof}[Proof of Lemma~\ref{lem:compute}]
    We argue that the algorithm runs in polynomial time.
    The rest of the statement follows from the fact that the algorithm terminates when it finds a set of almost tight pairs forming at most $k'$ solutions and that this can be used to construct structured pairs.

    In each iteration of our algorithm we can construct, in time $O(|X|^3)$, the set $X''$ of points whose $\alpha$-values we raise. Therefore, it remains to be shown that the number of iterations is bounded by $O(|X|^4)$.
    First, notice that $\delta \geq \frac{1}{|X|}\mu$ because in each iteration the increase in the variables by $\delta$ always leads at least one non-almost tight pair to become tight.
    Second, observe that $\OPT'_{LP}$ can be bounded from above by $k'\cdot r_{1/\epsilon}$ implying that also the optimal value of the dual is bounded by this value.
    Now, consider some iteration $\ell$ of the algorithm. Suppose that the current set of almost tight pairs forms more than $k'$ components. Then, we are guaranteed that $X''$ contains more than $k'$ points since we include at least one point per component. Therefore, after raising $\lambda$ and $\alpha_j$ for each $j \in X''$, the objective value of the dual incresed by at least $\delta \geq \frac{1}{|X|}\mu = \frac{1}{|X|^3} r_{1/\epsilon}$. By the upper bound on $\OPT'_{LP}$ we have that the number of iterations before finding a set of almost tight pairs forming at most $k'$ components is bounded by $O(|X|^4)$.
\end{proof}

\subsection{\label{subsec:Rounding-well-structured} Computing a feasible solution.}

Suppose we are given a solution $\left(\lambda,\alpha\right)$ to $(D)$ together with a set of structured pairs $\calb'\subseteq \calb$.
We describe an algorithm that computes a feasible solution whose cost is at most $3\cdot\OPT'+O(\epsilon)\cdot\OPT$.
This yields the proof of Lemma~\ref{lem:no-outliers-final-cost-bound}.

Recall that $\calb'$ has at most $k'$ components.
For each component $C=\left\{ (i_{1},r_{1}),...,(i_{\ell},r_{\ell})\right\} $ we want to do the following: We define one center $i\in X$ and assign it a radius $r$ such that $B(i,r)$ contains all points that are covered by the pairs in $C$, i.e., such that $X'(C):=\bigcup_{\ell'=1}^{\ell}B(i_{\ell'},r_{\ell'})\subseteq B(i,r)$.
If we assign a radius of $r$ to $i$, this yields a cost of $r$.
We want to ensure that $r\le3\cdot\sum_{j\in X'(C)}\alpha_{j}$, the intuition being that the~$\alpha_{j}$-values of the points in $X(C)$ pay for $r$.

However, if we do this for each component of $\calb'$, this	is not enough since this would yield only an upper bound of $3\cdot\sum_{j\in X'}\alpha_{j}$ of our total cost, but we need essentially an upper bound of $3\cdot\left(\sum_{j\in X'}\alpha_{j}-k'\cdot\lambda\right)$	to compare our cost to the value of the dual objective.
To this end, we identify a set $C_d \subseteq C$ for which we show an improved bound of $r \le 3 \cdot \sr(C_d) \le 3 \cdot \left(\sum_{j \in X(C)}\alpha_{j}-\lambda\right) + \frac{3\epsilon}{|X|}\OPT$ where  for any set $C' \subseteq \calb$ we define $\sr(C') := \sum_{(i',r') \in C'} r'$.
The balls corresponding to the pairs in $C_d$ are pairwise disjoint, i.e., $B(i,r) \cap B(i',r') = \emptyset$ for each $(i,r), (i',r') \in C_d$.
This yields the improvement over the bound of $3\cdot\sum_{j\in X'}\alpha_{j}$ that we need.

\begin{restatable}{lem}{lemnormalcomponent}
    \label{lem:normal-component}
    Let $C$ be a component of a set of tight pairs $\calb'$.
    Then there is a pair $(i,r)$ with $i\in X$ and $r \in \R$ and a set $C_d \subseteq C$ whose corresponding balls are pairwise disjoint such that
    \begin{enumerate}[label=$\arabic*.$,ref={\mbox{\rm $\arabic*$}}, noitemsep, leftmargin=1.5cm]
        \item $X(C)\subseteq B(i,r)$, and
        \item $r \le 3 \cdot \sr(C_d) \le 3\cdot\left(\sum_{j\in X'(C)}\alpha_{j}-\lambda\right) + \frac{3\epsilon}{|X|}\OPT$.
    \end{enumerate}
\end{restatable}

The pair $(i,r)$ can be found in polynomial time by trying all pairs $(i',r')$ with $i' \in X$ and $r' = d(i,j)$ for some $j \in X'$ such that $X(C)\subseteq B(i',r')$; we select such a pair $(i',r')$ that minimizes~$r'$.

\begin{proof}[Proof sketch of Lemma~\ref{lem:normal-component}]
    Given a component $C$, we construct a bipartite graph $G$ with a \emph{pair vertex} for each pair $(i',r')\in C$ and a \emph{point vertex}for each point $p\in X'$ covered by $C$.
    In $G$, a pair vertex for a pair $(i',r')\in C$ is connected by an edge to the point vertex for a point $p$ if and only if $p\in B(i',r')$, and the length of this edge is $r'$.
    Let $v$ be a vertex of $G$ that minimizes the length of the longest shortest path to the other vertices in $G$.
    If $v$ is a point vertex for a point $p\in X'$, we define $i:=p$; if $v$ is a pair vertex for a pair $(i',r')\in C$, we define $i:=i'$.
    We define $r$ to be the smallest radius such that $X(C)\subseteq B(i,r)$.

    We need to prove that $r\le3\cdot\sr(C_{d})\le3\cdot\left(\sum_{j\in X'(C)}\alpha_{j}-\lambda\right) + \frac{3\epsilon}{|X|} \OPT$ for a suitable set $C_{d}\subseteq C$.
    Recall that all pairs in $C$ are almost tight; hence, for an almost tight pair $(i',r')\in C$ the points in $B(i',r')$ satisfy that $\sum_{j\in B(i',r')\cap X'}\alpha_{j}\ge r+\lambda-\mu$.
    The set $C_{d}$ is constructed via long paths in $G$.
    The lengths of the edges of $G$ correspond to radii of pairs $(i',r')\in C$.
    Thus, if we have a very long path, then it contains pair vertices for which the sum of their radii is also large.
    We select such pair vertices for which the corresponding balls are pairwise disjoint since we want to argue that
    \begin{align*}
        \sr(C_{d})
         & = \sum_{(i',r')\in C_{d}} r'
        \leq \sum_{(i',r')\in C_{d}} \left( \sum_{j\in B(i',r')\cap X'} \alpha_{j} - \lambda + \mu \right) \\
         & \leq \sum_{j\in X'(C)} \alpha_{j} - \lambda + |C_{d}| \mu
        \leq \sum_{j\in X'(C)} \alpha_{j} - \lambda + \frac{\epsilon}{|X|}\OPT.
    \end{align*}
    Requiring these balls to be pairwise disjoint makes the construction of $C_d$ challenging.
    We do this via a careful case distinction in which we ensure that $r'\le3\cdot\sr(C_{d})$ (see Appendix~\ref{apx:lemma7}).
\end{proof}

Intuitively, the proof of Lemma~\ref{lem:no-outliers-final-cost-bound} now follows from applying Lemma~\ref{lem:normal-component} to each connected component and using that~${r^* \le r_{1/\epsilon} \le \epsilon\cdot \OPT}$.

\begin{proof}[Proof of Lemma~\ref{lem:no-outliers-final-cost-bound}]
    Let $\mathcal{C}_1$ be the set of components of the computed set of almost tight pairs $\calb'$ such that the corresponding balls of pairs in $\mathcal{C}_1$ are disjoint from the special almost tight pair $(i^*,r^*)$ and $\mathcal{C}_2$ be the set of components of $\calb' \setminus (i^*,r^*)$ that contain points covered by $(i^*,r^*)$. Let $r_C$ be the radius we use to cover component $C$ based on Lemma~\ref{lem:normal-component}. Then, for $\mathcal{C}_1$, it holds that
    \[
        \sum_{C \in \mathcal{C}_1}r_C \stackrel{{\scriptsize \mathrm{(Lem.\ref{lem:normal-component})}}}{\leq} 3 \left( \sum_{C \in \mathcal{C}_1}\sum_{j \in X'(C)} \alpha_j - \lambda \right) + |\mathcal{C}_1| \frac{3 \epsilon}{|X|}\OPT \leq 3\left(\sum_{j \in X'(\mathcal{C}_1)} \alpha_j - |\mathcal{C}_1|\lambda \right) + 3 \epsilon \cdot \OPT.
    \]
    Now, consider $\mathcal{C}_2$ as well as the component formed by $\mathcal{C}_2$ and $(i^*,r^*)$.
    We denote this component by $C^*$.
    Let $C^*_d$ denote a subset of $C^*$ whose corresponding balls are pairwise disjoint and maximizes~$\sr(C_d)$, and, for each $C \in \mathcal{C}_2$ let $C_d$ denote a subset of $C$ whose corresponding balls are pairwise disjoint and maximizes~$\sr(C_d)$.
    By Lemma~\ref{lem:normal-component}, we know that
    \[
        r_{C^*} \stackrel{{\scriptsize \mathrm{(Lem.\ref{lem:normal-component})}}}{\leq} 3 \cdot \sr(C^*_d).
    \]
    We will now show how to bound $\sr(C^*_d)$ from above using $\mathcal{C}_2$ and the almost tight pair $(i^*,r^*)$. To do this consider the following case distinction. First, assume that for each $C \in \mathcal{C}_2$ the corresponding balls of pairs in $C_d$ are disjoint from $B(i^*,r^*)$ as well.
    Then, since $\bigcup_{C \in \mathcal{C}_2}C_d \cup (i^*,r^*)$ is a set of pairwise disjoint almost tight pairs, we know that
    \[
        \sr(C^*_d) \leq \sum_{C \in \mathcal{C}_2}\sr(C) + \epsilon \cdot \OPT.
    \]
    Next, assume that $(i^*,r^*)$ is not disjoint from $C_d$ for some $C \in \mathcal{C}_2$. Then, either $\bigcup_{C \in \mathcal{C}_2}C_d$ is a set of disjoint almost tight pairs in $C^*$ maximizing the sum of radii or $(i^*,r*)$ is included in the set of disjoint almost tight pairs maximizing the sum of radii in $C^*$. In either case, we again have
    \[
        \sr(C^*_d) \leq \sum_{C \in \mathcal{C}_2}\sr(C_d) + \epsilon\cdot\OPT.
    \]
    By Lemma~\ref{lem:normal-component}, we have the following bound
    \[
        \sum_{C\in \mathcal{C}_2}\sr(C_d) \stackrel{{\scriptsize \mathrm{(Lem.\ref{lem:normal-component})}}}{\leq} \left( \sum_{C \in \mathcal{C}_2}\sum_{j \in X'(C)}  \alpha_j - \lambda \right) + |\mathcal{C}_2| \frac{\epsilon}{|X|}\OPT \leq \left(\sum_{j \in X'(\mathcal{C}_2)} \alpha_j - |\mathcal{C}_2|\lambda \right) + \epsilon \cdot \OPT.
    \]
    Therefore,
    \[
        r_{C^*} \leq 3 \cdot \sr(C^*_d) \leq 3 \left(\sum_{j \in X'(\mathcal{C}_2)} \alpha_j - |\mathcal{C}_2|\lambda \right) + 6 \epsilon \cdot \OPT.
    \]
    Finally, since each point $j \in X'$ is only covered by one of the components of $\mathcal{C}_1 \cup \mathcal{C}_2$ and $|\mathcal{C}_1|+|\mathcal{C}_2| > k$, we know that the final solution obtained by choosing a single ball to cover each component in $\mathcal{C}_1$ and a single ball to cover $C^*$ has cost
    \[
        \sum_{C \in \mathcal{C}_1}r_C + r_{C^*} \leq 3\left(\sum_{j \in X'} \alpha_j - k\lambda \right) + 9 \epsilon \cdot \OPT \leq 3 \cdot \OPT'_{LP} + O(\epsilon) \cdot \OPT.
    \]
    This concludes the proof.
\end{proof}

\section{$(3+\epsilon)$-approximation algorithm with outliers.}

In this section, we study the \textsc{Minimum Sum of Radii with Outliers} problem.
Like before, let $\epsilon>0$, assume w.l.o.g.~that $1/\epsilon$ is an integer, and we solve the problem exactly by enumeration in
time $n^{O(1/\epsilon)}$ if~${k\le1/\epsilon}$. Assuming that $k>1/\epsilon$, we guess the $1/\epsilon$ centers $i_{1},...,i_{1/\epsilon}$ from
$\OPT$ with largest radii and their corresponding radii $r_{1},...,r_{1/\epsilon}$, and solve the remaining problem with input points $X':=X\setminus\bigcup_{\ell=1}^{1/\epsilon}B(i_{\ell},r_{\ell})$, $k':=k-1/\epsilon$ centers, and up to $m$ outliers.
W.l.o.g., we assume that $|X'| > k' + m$.
Although a solution for the instance $(X',k',m)$ needs to cover only points in $X'$, we also allow solutions using centers on points in $X$.
Lemma~\ref{lem:outliers-biggest-radius} is analogous to the one for~Lemma~\ref{lem:biggest-radius}.

\begin{lem}
    \label{lem:outliers-biggest-radius}
    For the instance $(X',k',m)$ there exists an optimal solution with some centers $i_{1},...,i_{\ell}\in X$ and radii $r_{1},...,r_{\ell}$ such that $r_{\ell'}\le r_{1/\epsilon} \le\epsilon\cdot\OPT$ for each $\ell'\in[\ell]$.
\end{lem}
\begin{proof}
    For the instance $(X', k', m)$, let $\OPT'$ be an optimal solution.
    Let $\calb'$ denote the solution for $(X',k', m)$ that is obtained from $\OPT$ by disregarding the $1/\epsilon$ centers with largest radii and their corresponding radii.
    Assume for contradiction that $\calb'$ is not an optimal solution for $(X',k',m)$; then, $\OPT'$ together with these $1/\epsilon$ centers and corresponding radii form a feasible solution to $(X,k,m)$ of strictly smaller cost than $\OPT$, giving a contradiction.
    Thus, $\calb'$ has to be an optimal solution to $(X',k',m)$.
    Now, because the largest radius $r'$ of a center in $\calb'$ is smaller than the largest $r_1,\dots,r_{1/\epsilon}$ largest radii of centers in $\OPT$, it follows that
    \begin{align*}
        \OPT
         & \geq \sum_{i=1}^{1/\epsilon} r_i
        \geq \frac{1}{\epsilon} r_{1/\epsilon}
        \geq \frac{1}{\epsilon} r',
    \end{align*}
    completing the proof.
\end{proof}

Let $\OPT'$ denote the optimal solution to $(X',k',m)$.
As before, we can disregard radii larger than $r_{1/\epsilon}$, and thus we define $\calb:=\{(i,r):i\in X \wedge \exists j\in X' \text{ s.t. } r=d(i,j) \leq r_{1/\epsilon} \}$ which contains all pairs $(i,r)$ in $\OPT'$.
Observe that $|\calb| \in O(|X|^2)$.
We use the LP-relaxation $(P')$ stated below.
For each point $j\in X'$ the variable $y_{j}$ models whether $j$ is an outlier (and hence does not need to be covered by the selected pairs).
We define $(D')$ to be its dual, and we define $\OPT'_{LP}$ to be their optimal solution values.
Observe that $\OPT'_{LP} \leq \OPT' \leq \OPT$.

\vspace*{-14pt}
\begin{minipage}[t]{1\textwidth}
    \begin{minipage}[t]{0.45\textwidth}
        \begin{alignat*}{3}
            (P')\quad & \text{min}  & \displaystyle \sum_{(i,r)\in\calb} r \cdot x_{i,r} \hspace*{-13pt}               &         & \quad &                         \\
                      & \text{s.t.} & \displaystyle \sum_{(i,r) \in \calb:j\in B(i,r)} \hspace*{-17pt} x_{i,r} + y_{j} & \geq 1  &       & \forall j \in X'        \\
                      &             & \displaystyle \sum_{(i,r)\in\calb} x_{i,r}                                       & \leq k' &       &                         \\
                      &             & \displaystyle \sum_{j'\in X'} y_j                                                & \leq m  &       &                         \\
                      &             & x_{i,r}                                                                          & \geq 0  &       & \forall (i,r) \in \calb \\
                      &             & y_{j}                                                                            & \geq 0  &       & \forall j \in X'
        \end{alignat*}
    \end{minipage}
    \quad
    \begin{minipage}[t]{0.45\textwidth}
        \begin{alignat*}{3}
            (D')\quad & \text{max}  & \displaystyle \sum_{j\in X'}\alpha_{j}-k'\cdot\lambda -m\cdot\gamma \hspace*{-70pt} &                  & \quad &                        \\
                      & \text{s.t.} & \displaystyle \sum_{j\in B(i,r) \cap X'} \hspace*{-10pt} \alpha_{j}                 & \leq r + \lambda &       & \forall(i,r) \in \calb \\
                      &             & \alpha_{j}                                                                          & \leq \gamma      &       & \forall j\in X'        \\
                      &             & \alpha_{j}                                                                          & \geq 0           &       & \forall j\in X'        \\
                      &             & \lambda                                                                             & \geq 0           &       &                        \\
                      &             & \gamma                                                                              & \geq 0           &       &
        \end{alignat*}
    \end{minipage}
\end{minipage}
\medskip

Again, we are interested in a solution to the dual LP $(D')$ together with a set of tight pairs that is structured in a certain way.
However, since the setting with outliers is more complicated, we require more properties of such solutions.
Given a solution $\left(\lambda,\gamma,\alpha\right)$ to $(D')$, we say that a pair $(i,r) \in \calb$ is \emph{tight} if $\sum_{j\in B(i,r) \cap X'}  \alpha_{j} = r + \lambda$ and that a point $j \in X'$ is \emph{tight} if $\alpha_j = \gamma$.
For a set of pairs $\calb'$, let $\out(\calb') = X' \setminus X'(\calb')$ denote the set of point in $X'$ not covered by $\calb'$.
For an ordered set of pairs $\calb' = \{(i_1,r_1),(i_2,r_2),\dots\}$ we define $\calb'_q := \{(i_{q'}, r_{q'}) \in \calb' : q' \leq q \}$ for each $q \in \N$.

\begin{restatable}{defn}{defnoutliersorderlystructured}
    \label{defn:outliers-orderly-structured}
    Let $\calb' \subseteq \calb$ be an ordered set of pairs.
    We say that $\calb'$ is \emph{orderly structured} for a solution $(\lambda,\gamma,\alpha)$ to $(D')$ if there is a pair of indexes $(\ell,\ell')$ with $1 \leq \ell$ and $0 \leq \ell' \leq \ell$ and a pair $(i^*,r^*) \in \calb' \setminus \calb'_\ell$~and
    \begin{enumerate}[label=\normalfont OS$\arabic*.$,ref={\mbox{\rm OS$\arabic*$}}, noitemsep, leftmargin=1.5cm]
        \item \label{defn:os2} every point $j \in \out(\calb'_{\ell'})$ is tight and every pair $(i,r) \in \calb'$ is tight,
        \item \label{defn:os3} $|\out(\calb'_h)| > m \geq |\out(\calb'_{h'} \cup \{(i^*,r^*)\})|$ for every $h < \ell$ and $\ell' \leq h' \leq \ell$, and
        \item \label{defn:os4} $|\comp(\calb'_{\ell})| > k \geq |\comp(\calb'_{\ell'} \cup \{(i^*,r^*)\})|$.
    \end{enumerate}
\end{restatable}

We will show in Section~\ref{subsec:Computing-well-structured-outlier} that we can compute a solution to $(D')$ and a corresponding orderly structured set of pairs $\calb'$.
Our algorithm is completely different from the algorithm in Section~\ref{subsec:Computing-well-structured}, since the case with outliers is more complicated.

\begin{restatable}{lem}{lemoutlierscomputeso}
    \label{lem:outliers-compute-so}
    In polynomial time, we can compute a solution $\left(\lambda,\gamma,\alpha\right)$ to $(D')$ together with an orderly structured set of pairs $\calb'$ for $\left(\lambda,\gamma,\alpha\right)$.
\end{restatable}

Given a solution $(\lambda,\gamma,\alpha)$ to $(D')$ and a corresponding orderly structured set $\calb'$ we can compute a feasible solution to $(X',k',m)$ in polynomial time whose cost is essentially by at most a factor $3+\epsilon$ larger than the objective function value of $(\lambda,\gamma,\alpha)$ in $(D')$.
The core idea is, similarly as before, to compute a single pair $(i,r)$ for each connected component of tight balls, such that this pair covers all points covered by the component.
We~show that with an orderly structured set $\calb'$, we can compute a feasible solution from an unfeasible one whose cost can be bounded by the inclusion (and/or replacement) of a constant number of tight balls, increasing the cost by only a factor of $1+O(\epsilon)$.
Due to the outliers, this routine is more involved than the corresponding algorithm from~Section~\ref{subsec:Rounding-well-structured}.
\begin{restatable}{lem}{lemoutliersfinalcostbound}
    \label{lem:outliers-final-cost-bound}
    Let $\calb'$ be an orderly structured set for a solution $\left(\lambda,\gamma,\alpha\right)$ to $(D')$.
    In polynomial time we can compute a feasible solution whose cost is at most $3 \cdot \OPT'_{LP} + O(\epsilon) \cdot \OPT \leq 3 \cdot \OPT' + O(\epsilon) \cdot \OPT$.
\end{restatable}

The feasible solution due to Lemma~\ref{lem:outliers-final-cost-bound} and the guessed pairs $(i_{1},r_{1}),...,(i_{1/\epsilon},r_{1/\epsilon})$ then yield a solution to the original instance $(X,k,m)$ of cost at most $(3+O(\epsilon)) \cdot \OPT$.

\begin{thm}
    \label{thm:3+eps-outlier}
    For any $\epsilon>0$ there is a $(3+\epsilon)$-approximation algorithm for the \textsc{Minimum Sum of Radii with Outliers} problem.
\end{thm}

In Appendix~\ref{apx:generalized-lb} we extend this algorithm to a $(3.5+\epsilon)$-approximation algorithm for the setting of generalized lower bounds (with and without outliers).

\subsection{\label{subsec:Computing-well-structured-outlier}Computing a orderly structured set solution.}

We present our algorithm for computing a solution $\left(\lambda,\gamma,\alpha\right)$ to $(D')$ and corresponding orderly structured set $\calb'$.
This will yield the proof of Lemma~\ref{lem:outliers-compute-so}.

We first describe a subroutine within our algorithm.
This routine and variations of it have been used in prior work on \textsc{Minimum Sum of Radii} and its variants~\cite{Ahmadian2016,Charikar2004,Friggstad2022}.
We assume that we are given a fixed value for $\lambda$.
We initialize a set of pairs $\calp$ by $\calp = \emptyset$ and a solution to $(D')$ by $\alpha_{j}:=0$ for each $j\in X'$ and~$\gamma:=\infty$.
Our~subroutine is a continuous process that is subdivided into phases.
At the beginning of each phase, we define $X''$ to be the set of points that are not covered by pairs $(i,r) \in \calb$ that are tight in the current dual solution.
For each point $j\in X''$, we raise $\alpha_{j}$ simultaneously and uniformly at the same rate.
The phase ends when there is a pair $(i,r) \notin \calp$ that is tight in the current dual solution.
At the end of the phase, we add all such pairs to $\calp$, one at a time.
We assume that the set of pairs $\calb$ is totally ordered, and the order in which the tight pairs are added to $\calp$ follows this ordering.
We stop adding pairs to $\calp$ when there are at most $m$ points $j\in X'$ that are not covered by $\calp$. Finally, we define $\gamma:=\max_{j\in X'}\alpha_{j}$.
Notice that if a point $j \in X'$ is not covered by $\calp$ at the beginning of the last phase, then $\alpha_j$ is raised during the last phase. Thus, $j$ is tight when the algorithm terminates, even if $j$ is covered by $\calp$ at the end.
Also, given the final values for $\lambda$ and $\alpha_j$ for each $j\in X'$, our choice for $\gamma$ optimizes the dual objective function value (while satisfying all constraints of~$(D')$).

Each pair $(i,r)$ might become tight at some point in time in the routine above, depending on $\lambda$.
For each~$(i,r) \in \calb$, we define a function $f_{(i,r)}$ such that $f_{(i,r)}(\lambda)$ equals this point in time for each $\lambda$; if $(i,r)$ never becomes tight for some value of $\lambda$, then $f_{(i,r)}(\lambda)=\infty$. Also, we define a function $g_{1}:[0,\infty)\rightarrow\R$ which defines for each $\lambda$ the point in time when the first phase ends.
One can easily show that $g_{1}$ is the lower envelope (i.e., the point-wise minimum), of the functions $f_{(i,r)}(\lambda)$.

\begin{restatable}{lem}{lemoutliersaffinefuncg}
    \label{lem:outliers-affine-func-g}
    For each $\lambda \ge 0$ we have that $g_{1}(\lambda)=\min_{(i,r) \in \calb}f_{(i,r)}(\lambda)$.
    Also, $g_{1}$ is a piecewise affine function with at most $O(|\calb|)$ many pieces.
    Moreover, we can compute $g_{1}$ in polynomial time.
\end{restatable}
\begin{proof}
    Let $(i,r) \in \calb$ be a pair whose corresponding ball $B(i,r)$ contain at least one point in $X'$ and consider its respective constraint in $(D')$.
    We define the affine function $h_{(i,r)}(\lambda) = \frac{r + \lambda}{|B(i,r) \cap X'|}$.
    Notice that $h_{(i,r)}(\lambda)$ measures the time that $(i,r)$ would take to become tight if there was no other constraints in $(D')$.
    Also, notice that if $(i,r)$ becomes tight for some $\lambda \geq 0$ in the algorithm when the first phase ends, then $h_{(i,r)}(\lambda) = f_{(i,r)}(\lambda)$.
    Therefore, the lower envelope of the functions $h_{(i,r)}(\lambda)$ is the same as the lower envelope of the functions $f_{(i,r)}(\lambda)$, hence $g_1(\lambda) = \min_{(i,r) \in \calb} f_{(i,r)}(\lambda) = \min_{(i,r) \in \calb} h_{(i,r)}(\lambda)$.
    Then, by applying standard polynomial time line crossing algorithms for every pair of affine functions $h_{(i,r)}(\lambda)$ and $h_{(i',r')}(\lambda)$ one can compute the lower envelope $g_1$.
    Furthermore, since each $h_{(i,r)}(\lambda)$ is affine, it follows that $g_1(\lambda)$ is piecewise affine and that the number of pieces of $g_1(\lambda)$ is bounded by the number of pairs in $\calb$.
\end{proof}

We compute $g_{1}$ and identify a polynomial number of intervals $(b_{0},b_{1}),\dots,(b_{t-1},b_t)$ with $b_{0}=0$ and $b_{t} > 2 \cdot |X'| \cdot k' \cdot r_{1/\epsilon}$ such that $g_{1}$ is an affine function when we restrict it to any of these intervals.

We now do a binary search over the values $b_0, b_1 ,...,b_t$.
For each candidate value $b_{\ell}$, we run our subroutine with~${\lambda = b_{\ell}}$ and we compute the number of components in $\calp$.
The proof for the following Lemma~\ref{lem:outliers-small-big-lambda} is similar to a proof from Ahmandian and Swamy~\cite{Ahmadian2016}.

\begin{lem}
    \label{lem:outliers-small-big-lambda}
    If we run our subroutine with $\lambda=b_{0} = 0$, then $|\comp(\calp)| > k'$.
    If we run it with ${\lambda = b_{t} > 2 \cdot |X'| \cdot k' \cdot r_{1/\epsilon}}$, then $|\comp(\calp)| \leq k'$.
\end{lem}
\begin{proof}
    If the subroutine is run with $\lambda = 0$, then at the first iteration every pair $(i,0)$ for $i \in X'$ becomes tight and it stops.
    Hence the number of components of tight pairs is $|X'|$.

    Assume that the subroutine is run with $\lambda > 2 \cdot |X'| \cdot k' \cdot r_{1/\epsilon}$.
    Let $(\lambda, \gamma, \alpha)$ be the computed solution to~$(D')$ and notice that for any tight pair $(i,r)$, it follows that
    \begin{align}
        \label{lem:outliers-small-big-lambda:ineq-1}
        \frac{\lambda}{|X'|}
         & = \frac{1}{|X'|} \left( \sum_{j \in B(i,r) \cap X'} \alpha_j - r \right)
        \leq \frac{1}{|X'|} \left( |X'| \cdot \gamma - r \right)
        \leq \gamma.
    \end{align}

    Let $\calb'$ denote the set of pairs that are tight picked by the algorithm and let $(i^*,r^*)$ be the last pair picked, and for each component $C \in \comp(\calb')$, let $(i_C,r_C)$ denote a pair in $C$.
    From the stop condition, we have that $|\out(\calb' \setminus \{i^*,r^*\})| > m$, and, from the definition of the algorithm, we also have that $\alpha_j = \gamma$ for every $j \in \out(\calb' \setminus \{i^*,r^*\})$.
    Let $M$ denote a set of exactly $m$ points such that $\out(\calb') \subseteq M \subsetneq \out(\calb' \setminus \{i^*,r^*\})$ and let $j' \in \out(\calb' \setminus \{i^*,r^*\}) \setminus M$, and recall that from Lemma~\ref{lem:outliers-biggest-radius} we have that no radius in $\OPT'_{LP}$ is larger than $r_{1/\epsilon}$.
    Now, assume, for a contradiction, that by the end of the $|\comp(\calb')| > k'$.
    Then, because removing a pair from $\calb'$ can decrease the number of components by at most one, it follows that
    \begin{align}
        \label{lem:outliers-small-big-lambda:ineq-2}
        |\comp(\calb' \setminus \{i^*,r^*\})|
         & \geq |\comp(\calb')| - 1
        \geq k'.
    \end{align}
    Then,
    \begin{align*}
        \frac{\lambda}{2 |X'|}
         & \stackrel{\phantom{\scriptsize \mathrm{(ineq.\ref{lem:outliers-small-big-lambda:ineq-2})}}}{>} k' \cdot r_{1/\epsilon}
        \stackrel{{\scriptsize \mathrm{(Lem.\ref{lem:outliers-biggest-radius})}}}{\geq} \OPT'_{LP}
        \stackrel{\phantom{\scriptsize \mathrm{(ineq.\ref{lem:outliers-small-big-lambda:ineq-2})}}}{\geq} \sum_{j \in X'} \alpha_j - m \cdot \gamma - k' \cdot \lambda                                           \\
         & \stackrel{\phantom{\scriptsize \mathrm{(ineq.\ref{lem:outliers-small-big-lambda:ineq-2})}}}{=} \sum_{j \in X' \setminus M} \alpha_j + \sum_{j \in M} \alpha_j - m \cdot \gamma - k' \cdot \lambda     \\
         & \stackrel{\phantom{\scriptsize \mathrm{(ineq.\ref{lem:outliers-small-big-lambda:ineq-2})}}}{=} \sum_{j \in X' \setminus M} \alpha_j - k' \cdot \lambda
        \stackrel{\phantom{\scriptsize \mathrm{(ineq.\ref{lem:outliers-small-big-lambda:ineq-2})}}}{\geq} \sum_{j \in X' \setminus \out(\calb' \setminus \{i^*,r^*\})} \alpha_j + \alpha_{j'} - k' \cdot \lambda \\
         & \stackrel{\phantom{\scriptsize \mathrm{(ineq.\ref{lem:outliers-small-big-lambda:ineq-2})}}}{\geq} \sum_{\substack{j \in B(i_C,r_C)                                                                    \\ C \in \comp(\calb' \setminus \{i^*,r^*\})}} \alpha_j + \gamma - k' \cdot \lambda
        \stackrel{\phantom{\scriptsize \mathrm{(ineq.\ref{lem:outliers-small-big-lambda:ineq-2})}}}{\geq} |\comp(\calb' \setminus \{i^*,r^*\})| \cdot \lambda + \gamma - k' \cdot \lambda                        \\
         & \stackrel{{\scriptsize \mathrm{(ineq.\ref{lem:outliers-small-big-lambda:ineq-2})}}}{\geq} \gamma
        \stackrel{{\scriptsize \mathrm{(ineq.\ref{lem:outliers-small-big-lambda:ineq-1})}}}{\geq} \frac{\lambda}{|X'|},
    \end{align*}
    which is a contradiction.
\end{proof}

A consequence of Lemma~\ref{lem:outliers-small-big-lambda} is that we can find a value $b_{\ell^{*}}$ such that
\begin{enumerate}[noitemsep, leftmargin=1.5cm]
    \item for $\lambda=b_{\ell^{*}}=:b_{L}^{(2)}$, the subroutine computes a set $\calp$ such that $|\comp(\calp)| > k'$, but
    \item for $\lambda=b_{\ell^{*}+1}=:b_{R}^{(2)}$, the subroutine computes a set $\calp$ such that $|\comp(\calp)| \leq k'$.
\end{enumerate}
We continue iteratively maintaining certain invariants.
At the beginning of each iteration $s$, we assume that we are given an interval $(b_{L}^{(s)},b_{R}^{(s)})$ and a set of pairs $\calb_{s'}$ for each $s'< s$.
Also, we assume that if we run our subroutine with $\lambda = b_L^{(s)}$ we compute a set of pairs $\calp$ with more than $k$ components, and if we run our subroutine with $\lambda = b_R^{(s)}$ the computed set $\calp$ has at most $k$ components.
Furthermore, we assume that if we run our subroutine with any parameter $\lambda \in (b_{L}^{(s)},b_{R}^{(s)})$, for each $s'< s$ the pairs in $\calb_{s'}$ become tight during phase $s'$.
Notice that all these invariants are satisfied after the first iteration, i.e., for $s=2$.

We describe now our routine for iteration $s$.
We define a function $g_{s}: (b_{L}^{(s)},b_{R}^{(s)}) \rightarrow \mathbb{R}$ which defines for each $\lambda \in (b_{L}^{(s)},b_{R}^{(s)})$ the point in time when phase $s$ ends if we execute our subroutine with parameter $\lambda$.

\begin{restatable}{lem}{lemoutliersaffinefuncgs}
    \label{lem:outliers-affine-func-gs}
    The function $g_{s}: (b_{L}^{(s)},b_{R}^{(s)}) \rightarrow \mathbb{R}$ is a piecewise affine function with at most $O(|\calb|)$ many pieces.
    Moreover, we can compute $g_{s}$ in polynomial time.
\end{restatable}
\begin{proof}
    This proof is similar to the proof of Lemma~\ref{lem:outliers-affine-func-g} (but more involved).
    For convenience, define $\calb_0 = \emptyset$ and $g_{0}(\lambda) = 0$, as there is no phase $0$.
    We show by induction that at each iteration $s'$, for each $s'' < s'$ the pairs in $\calb_{s''}$ become tight at the end of phase $s''$ when executing the subroutine for $\lambda \in (b_{L}^{(s')},b_{R}^{(s')})$ and that $g_{s''}(\lambda)$ is affine in $(b_{L}^{(s')},b_{R}^{(s')})$.

    For $s' = 0$, we have that $g_{0}(\lambda)$ is affine but no pair has become tight prior to the first phase.
    For some $s'$, assume that the inductive hypothesis holds.
    For each $s'' \leq s'$, let $X''_{s''}$ denote the set of points that are not covered by pairs in $\calb_{s''-1}$ during phase~$s''$, and let $(i,r) \in \calb$ be a pair whose ball corresponding ball $B(i,r)$ contain at least one point in $X''_{s''}$.
    Observe that since $\calb_{s''}$ are fixed for every $s'' < s'$, it follows that $X''_{s''}$ is fixed for every $s'' \leq s'$.
    Because for every $s'' < s'$ we have that $g_{s''}(\lambda)$ is affine for $\lambda \in (b_L^{(s')},b_R^{(s')})$, we have that $\Delta_{s''}(\lambda) = g_{s''}(\lambda) - g_{s''-1}(\lambda)$ is also affine and that this function measures the time between start and end of phase $s''$.
    Now, consider the respective constraint in $(D')$ of the pair $(i,r)$, and notice that at phase $s''$, the slack in the constraint reduces in exactly $|B(i,r) \cap X''_{s''}| \Delta_{s''}(\lambda)$.
    Thus, we can define the affine function $h_{(i,r)}(\lambda) = \frac{r + \lambda - \sum_{s'' < s'} |B(i,r) \cap X''_{s''}| \Delta_{s''}(\lambda)}{|B(i,r) \cap X''_{s'}|}$ which measures the time that $(i,r)$ would need to become tight if we disregarded the other constraints in $(D')$ during phase $s'$.
    Therefore, the lower envelope of the functions $h_{(i,r)}(\lambda)$ is exactly $g_{s'}(\lambda)$, i.e., $g_{s'}(\lambda) = \min_{(i,r) \in \calb} h_{(i,r)}(\lambda)$, and it is piecewise affine in $\in (b_L^{(s')},b_R^{(s')})$ and the number of its affine pieces is bounded by the number of pairs in $\calb$.
    By applying standard polynomial time line crossing algorithms for every pair of affine functions $h_{(i,r)}(\lambda)$ and $h_{(i',r')}(\lambda)$ one can compute the function $g_{s'}(\lambda)$.
    Then, as the algorithm defines $b_{L}^{(s'+1)},b_{R}^{(s'+1)}$ to be one of the affine pieces of $g_{s'}(\lambda)$, it follows that for each $s'' < s'+1$ the pairs in $\calb_{s''}$ become tight at the end of phase $s''$ and that $g_{s''}(\lambda)$ is affine in $(b_{L}^{(s'+1)},b_{R}^{(s'+1)})$.
    This completes the induction.

    Because for each $s'' < s$ the same pairs in $\calb_{s''}$ become tight at the end of phase $s''$ when executing the subroutine for $\lambda \in (b_{L}^{(s)},b_{R}^{(s)})$ and $g_{s''}(\lambda)$ is affine in $(b_{L}^{(s)},b_{R}^{(s)})$, with the same arguments as in the inductive step we conclude that $g_{s}(\lambda)$ is piecewise affine in $(b_{L}^{(s)},b_{R}^{(s)})$ and with at most $|\calb|$ affine~pieces.
\end{proof}

Like before, we compute $g_{s}$ and identify a polynomial number of intervals $(b_{0}^{(s)},b_{1}^{(s)}),\dots,(b_{t'-1}^{(s)},b_{t'}^{(s)})$ with $b_{0}^{(s)}=b_{L}^{(s)}$ and $b_{t'}^{(s)}=b_{R}^{(s)}$ such that $g_{s}$ is an affine function when we restrict it to any of the intervals.
We do a binary search to find a value $b_{\ell^{*}}^{(s)}$ such that
\begin{enumerate}[noitemsep, leftmargin=1.5cm]
    \item for $\lambda=b^{(s)}_{\ell^{*}}=:b_{L}^{(s+1)}$, the subroutine computes a set $\calp$ with $|\comp(\calp)| > k'$, but
    \item for $\lambda=b^{(s)}_{\ell^{*}+1}=:b_{R}^{(s+1)}$, the subroutine computes a set $\calp$ with $|\comp(\calp)| \leq k'$.
\end{enumerate}

Hence, at the beginning of the next iteration $s+1$ we satisfy our invariants mentioned above.
Let $s^{*}$ denote the iteration in which we are given an interval $(b_{L}^{(s^{*})},b_{R}^{(s^{*})})$ such that if we execute our subroutine for any $\lambda \in (b_{L}^{(s^{*})},b_{R}^{(s^{*})})$, the routine ends at the end of phase $s^*-1$ and $g_{s^*-1}: (b_{L}^{(s^*)},b_{R}^{(s^*)}) \rightarrow \mathbb{R}$ is an affine function.

\begin{restatable}{lem}{lemoutliersalgends}
    \label{lem:outliers-alg-ends}
    We can check in polynomial time whether iteration $s^{*}$ has been reached.
    Also, $s^{*} \leq |\calb|$.
\end{restatable}
\begin{proof}
    Let $g_{s-1}(\lambda)$ be the piecewise affine function on $(b_{L}^{(s-1)},b_{R}^{(s-1)})$ and consider the interval $(b_{L}^{(s)},b_{R}^{(s)})$ chosen by the algorithm on which $g_{s-1}(\lambda)$ is affine.
    If the subroutine ends at phase $s$ for some $\lambda \in (b_{L}^{(s)},b_{R}^{(s)})$, then the same pairs are tight at the end of phase $s$ for any $\lambda \in (b_{L}^{(s)},b_{R}^{(s)})$.
    Recall that these tight pairs are then added to the set $\calp$, one at a time, and that the subroutine stops if there are at most $m$ points in $X'$ not covered by $\calp$, and also recall that these tight pairs are always considered in the same fixed order to be added to $\calp$.
    Therefore, if the subroutine ends at phase at this phase, we conclude that $s^* = s$.
    Also, since there cannot be more phases than the number of possible pairs that could become tight, it follows that $s^{*}$ is upper bounded by the total number of pairs in $\calb$.
\end{proof}

Our invariant implies that if we run our subroutine with $\lambda=b_{L}^{(s^{*})}$ we compute a set of pairs $\calp_{L}$ with more than $k$ components, and if we run it with $\lambda=b_{R}^{(s^{*})}$ we compute a set of pairs $\calp_{R}$ with at most $k$ components.
Also, if we run the subroutine for some $\lambda\in(b_{L}^{(s^{*})},b_{R}^{(s^{*})})$ it computes a set $\calp_{M}$, and this is the same set for each $\lambda\in(b_{L}^{(s^{*})},b_{R}^{(s^{*})})$.
Suppose that $\calp_{M}$ has at most $k$ components.
We set $\lambda:=b_{L}^{(s^{*})}+\mu$ for a small $\mu>0$.
With a continuity argument, we can show that every pair in $\calp_{M}$ and every point not covered by $\calp_{M}$ are essentially tight in the execution of the subroutine that yields $\calp_{L}$, assuming that $\mu$ is sufficiently small.
We devise a routine that carefully mixes the pairs from $\calp_{M}$ and $\calp_{L}$ and that computes an orderly structured set $\calb'$.
A similar reasoning can be used if $\calp_{M}$ has more than most $k$ components.

\begin{restatable}{lem}{lemoutliersfindospolytime}
    \label{lem:outliers-find-os-polytime}
    Given $b_{L}^{(s^{*})}, b_{R}^{(s^{*})}$, in polynomial time we can compute a value $\lambda \in \{b_{L}^{(s^{*})}, b_{R}^{(s^{*})}\}$ and a set $\calb'$ such that $\calb'$ is orderly structured for a solution $(\lambda,\gamma,\alpha)$ to $(D')$.
\end{restatable}

The lemmas above yield the proof of Lemma~\ref{lem:outliers-compute-so}.
\begin{proof}[Proof of Lemma~\ref{lem:outliers-compute-so}]
    Lemma~\ref{lem:outliers-compute-so} is a direct consequence of Lemmas~\ref{lem:outliers-affine-func-g}, \ref{lem:outliers-small-big-lambda}, \ref{lem:outliers-affine-func-gs}, \ref{lem:outliers-alg-ends}, and~\ref{lem:outliers-find-os-polytime}.
\end{proof}

\bibliographystyle{abbrv}
\bibliography{Arxiv-k-sum-radii.bib}

\appendix

\section{Barrier of $3$}
\label{apx:tight_example}

In this section, we prove that our analysis of the approximation factor of $3+\epsilon$ due to Theorem~\ref{thm:3+eps} is tight.
For any given $\eta>0$, we give an instance in which the constructed dual solution is by a factor of $3-\eta$ cheaper than the optimal solution.
Thus, with our primal-dual analysis, one cannot prove a better approximation ratio than $3-\eta$.
In particular, we argue that the same happens in previous primal-dual approaches for the problem~\cite{Ahmadian2016,Charikar2004}.
Thus, our approximation ratio of $3+\epsilon$ is also a limit for them.

We construct a family of instances for any $k\in\N$.
The construction also depends on a parameter $h \in \N_{\geq 3}$ which will make sure that for any given $\eta > 0$ we find an instance for which the constructed dual solution is by a factor of $3- \eta$ cheaper than the optimum solution.
We define a metric space $(X^h_{k},d)$ based on a graph $G^h_k$ with node set $X^h_k$ together with unit edge lengths, where the distance $d(i,j)$ between two points $i, j \in X^h_k$ is given by the shortest $i$-$ j$-path in $G^h_k$ (if there is no $i$-$j$-path, the distance is infinitely large).
To this end, let $G^h_k$ be the graph consisting of $k$ disjoint copies of the graph with node set $\{ v_1, \dots, v_h \} \cup \{v_{ij} : i,j \in [h]\}$ and edge set $\{ \{v_k, v_{ij}\} : i,j,k \in [h], k \neq i\}$.
Figure \ref{fig:graph G^h_k} depicts the graph $G^h_k$ and therefore visualizes the metric space on $X^h_k$.

\begin{figure}[htb]
    \centering
    \begin{tikzpicture} 

        \path[-]
        (0,0) edge[line width=0.5pt] (2,4.5)
        (0,0) edge[line width=0.5pt] (2,4.2)
        (0,0) edge[line width=0.5pt] (2,3.5)
        (0,0) edge[line width=0.5pt] (2,2.5)
        (0,0) edge[line width=0.5pt] (2,2.2)
        (0,0) edge[line width=0.5pt] (2,1.5)
        (0,2) edge[line width=0.5pt] (2,4.5)
        (0,2) edge[line width=0.5pt] (2,4.2)
        (0,2) edge[line width=0.5pt] (2,3.5)
        (0,2) edge[line width=0.5pt] (2,0)
        (0,2) edge[line width=0.5pt] (2,-0.3)
        (0,2) edge[line width=0.5pt] (2,-1)
        (0,3.5) edge[line width=0.5pt] (2,2.5)
        (0,3.5) edge[line width=0.5pt] (2,2.2)
        (0,3.5) edge[line width=0.5pt] (2,1.5)
        (0,3.5) edge[line width=0.5pt] (2,0)
        (0,3.5) edge[line width=0.5pt] (2,-0.3)
        (0,3.5) edge[line width=0.5pt] (2,-1)
        ;

        \fill (0,0) circle (3pt) node[left,color=black] {$v_h$};
        \node[color=black]  at (0,1) {$\vdots$};
        \fill (0,2) circle (3pt) node[left,color=black] {$v_2$};
        \fill (0,3.5) circle (3pt) node[left,color=black] {$v_1$};
        \fill (2,4.5) circle (3pt) node[right,color=black] {$v_{11}$};
        \fill (2,4.2) circle (3pt) node[right,color=black] {$v_{12}$};
        \node[color=black]  at (2,3.95) {$\vdots$};
        \fill (2,3.5) circle (3pt) node[right,color=black] {$v_{1h}$};
        \fill (2,2.5) circle (3pt) node[right,color=black] {$v_{21}$};
        \fill (2,2.2) circle (3pt) node[right,color=black] {$v_{22}$};
        \node[color=black]  at (2,1.95) {$\vdots$};
        \fill (2,1.5) circle (3pt) node[right,color=black] {$v_{2h}$};
        \node[color=black]  at (2,1) {$\vdots$};
        \fill (2,0) circle (3pt) node[right,color=black] {$v_{h1}$};
        \fill (2,-0.3) circle (3pt) node[right,color=black] {$v_{h2}$};
        \node[color=black]  at (2,-0.55) {$\vdots$};
        \fill (2,-1) circle (3pt) node[right,color=black] {$v_{hh}$};

        \path[-]
        (4,0) edge[line width=0.5pt] (6,4.5)
        (4,0) edge[line width=0.5pt] (6,4.2)
        (4,0) edge[line width=0.5pt] (6,3.5)
        (4,0) edge[line width=0.5pt] (6,2.5)
        (4,0) edge[line width=0.5pt] (6,2.2)
        (4,0) edge[line width=0.5pt] (6,1.5)
        (4,2) edge[line width=0.5pt] (6,4.5)
        (4,2) edge[line width=0.5pt] (6,4.2)
        (4,2) edge[line width=0.5pt] (6,3.5)
        (4,2) edge[line width=0.5pt] (6,0)
        (4,2) edge[line width=0.5pt] (6,-0.3)
        (4,2) edge[line width=0.5pt] (6,-1)
        (4,3.5) edge[line width=0.5pt] (6,2.5)
        (4,3.5) edge[line width=0.5pt] (6,2.2)
        (4,3.5) edge[line width=0.5pt] (6,1.5)
        (4,3.5) edge[line width=0.5pt] (6,0)
        (4,3.5) edge[line width=0.5pt] (6,-0.3)
        (4,3.5) edge[line width=0.5pt] (6,-1)
        ;

        \fill (4,0) circle (3pt) node[left,color=black] {};
        \node[color=black]  at (4,1) {$\vdots$};
        \fill (4,2) circle (3pt) node[left,color=black] {};
        \fill (4,3.5) circle (3pt) node[left,color=black] {};
        \fill (6,4.5) circle (3pt) node[right,color=black] {};
        \fill (6,4.2) circle (3pt) node[right,color=black] {};
        \node[color=black]  at (6,3.95) {$\vdots$};
        \fill (6,3.5) circle (3pt) node[right,color=black] {};
        \fill (6,2.5) circle (3pt) node[right,color=black] {};
        \fill (6,2.2) circle (3pt) node[right,color=black] {};
        \node[color=black]  at (6,1.95) {$\vdots$};
        \fill (6,1.5) circle (3pt) node[right,color=black] {};
        \node[color=black]  at (6,1) {$\vdots$};
        \fill (6,0) circle (3pt) node[right,color=black] {};
        \fill (6,-0.3) circle (3pt) node[right,color=black] {};
        \node[color=black]  at (6,-0.55) {$\vdots$};
        \fill (6,-1) circle (3pt) node[right,color=black] {};

        \node[color=black]  at (8,1.75) {$\cdots$};

        \path[-]
        (10,0) edge[line width=0.5pt] (12,4.5)
        (10,0) edge[line width=0.5pt] (12,4.2)
        (10,0) edge[line width=0.5pt] (12,3.5)
        (10,0) edge[line width=0.5pt] (12,2.5)
        (10,0) edge[line width=0.5pt] (12,2.2)
        (10,0) edge[line width=0.5pt] (12,1.5)
        (10,2) edge[line width=0.5pt] (12,4.5)
        (10,2) edge[line width=0.5pt] (12,4.2)
        (10,2) edge[line width=0.5pt] (12,3.5)
        (10,2) edge[line width=0.5pt] (12,0)
        (10,2) edge[line width=0.5pt] (12,-0.3)
        (10,2) edge[line width=0.5pt] (12,-1)
        (10,3.5) edge[line width=0.5pt] (12,2.5)
        (10,3.5) edge[line width=0.5pt] (12,2.2)
        (10,3.5) edge[line width=0.5pt] (12,1.5)
        (10,3.5) edge[line width=0.5pt] (12,0)
        (10,3.5) edge[line width=0.5pt] (12,-0.3)
        (10,3.5) edge[line width=0.5pt] (12,-1)
        ;

        \fill (10,0) circle (3pt) node[left,color=black] {};
        \node[color=black]  at (10,1) {$\vdots$};
        \fill (10,2) circle (3pt) node[left,color=black] {};
        \fill (10,3.5) circle (3pt) node[left,color=black] {};
        \fill (12,4.5) circle (3pt) node[right,color=black] {};
        \fill (12,4.2) circle (3pt) node[right,color=black] {};
        \node[color=black]  at (12,3.95) {$\vdots$};
        \fill (12,3.5) circle (3pt) node[right,color=black] {};
        \fill (12,2.5) circle (3pt) node[right,color=black] {};
        \fill (12,2.2) circle (3pt) node[right,color=black] {};
        \node[color=black]  at (12,1.95) {$\vdots$};
        \fill (12,1.5) circle (3pt) node[right,color=black] {};
        \node[color=black]  at (12,1) {$\vdots$};
        \fill (12,0) circle (3pt) node[right,color=black] {};
        \fill (12,-0.3) circle (3pt) node[right,color=black] {};
        \node[color=black]  at (12,-0.55) {$\vdots$};
        \fill (12,-1) circle (3pt) node[right,color=black] {};

        \draw [decorate,decoration={brace,amplitude=10pt, mirror},xshift=-4pt,yshift=0pt]
        (0,-1.5) -- (12.2,-1.5 ) node [black,midway,yshift=-0.6cm] {$k$ copies};

    \end{tikzpicture}
    \caption{Graph $G^h_k$ with unit edge lengths inducing a metric space on $X^h_k$}
    \label{fig:graph G^h_k}
\end{figure}
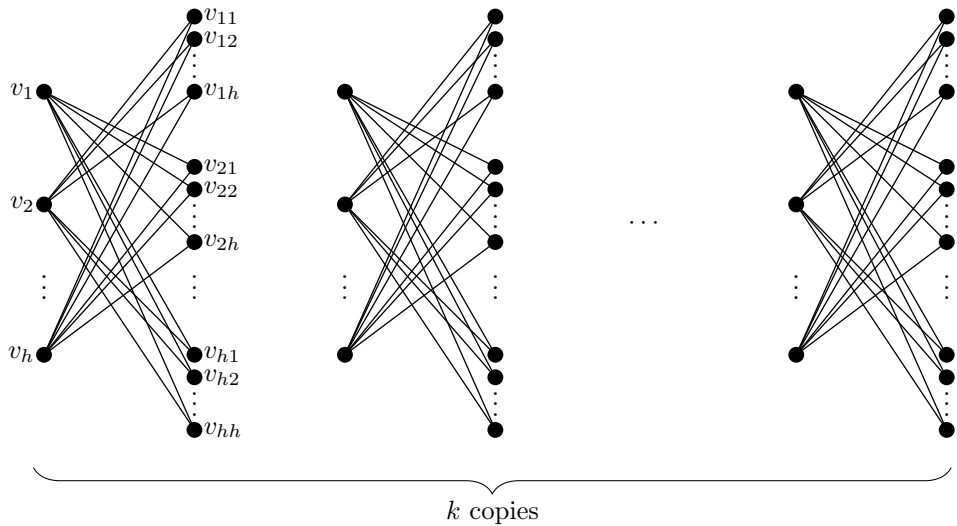

Let $(X^h_{k},k)$ denote the resulting instance of \textsc{Minimum Sum of Radii}.

\begin{lem}
    \label{lem: opt of X^h_k instance}
    For any $k\in\N$ and $h \in \N_{\geq 3}$, the value of the optimal solution to $(X^h_{k},k)$ is $3k$.
\end{lem}
\begin{proof}
    Since the distance between points in different connected components of $G^h_k$ is infinitely large and $G^h_k$ has $k$ connected components, the optimum solution has to place one center in each of these $k$ components of $G^h_k$.
    As all edges have unit length, choosing in each component of $G^h_k$ any of the points as center gives a corresponding radius of $3$.  Hence, the optimum solution consists of $k$ clusters of radius $3$ each, yielding a value of $3k$.
\end{proof}

We say that a pair $(i,r)$ is \emph{tight} if $\sum_{j\in B(i,r)\cap X'}\alpha_{j} = r+\lambda$.
In the following, we consider a variant of the algorithm from Section~\ref{subsec:Computing-well-structured} that requires the stronger assumption of tight pairs.
In this variant, we greedily select a maximal set $X''\subseteq X'$ such that each tight pair $(i,r)$ covers at most one point in $X''$, and, we increase the variables until some pair becomes tight, i.e., at each iteration we increase the variables by
\[
    \delta:=\min_{(i,r) \in \calb} \left\{ \frac{r + \lambda - \sum_{j \in B(i,r)\cap X'} \alpha_j}{|B(i,r) \cap X''|-1} : |B(i,r) \cap X''|\geq 2 \wedge (i,r) \text{ is not tight} \right\}.
\]

When we execute our algorithm on $(X^h_{k},k)$, we raise $\lambda$ and the $\alpha_{j}$ variables until one pair $(i,1)$ for some $i\in X^h_k$ becomes tight (details are given in the proof of Theorem~\ref{thm:tight_factor_of_3}).
In fact, at the same time, many such pairs become tight, and we obtain exactly $k$ components in the sense of Definition~\ref{defn:components}.
Thus, the algorithm stops.
However, at this point, the value of the dual solution is only $\left(k + \frac{k (2h-1)}{h^2-h}\right)$, which is essentially $3$ times smaller than $\OPT$, for $h$ large enough.
In our analysis, we get a slightly better approximation ratio since we first guessed the $1/\epsilon$ largest radii from $\OPT$.
However, if $k$ is sufficiently large compared to $1/\epsilon$, the gain is negligibly small.

\begin{thm}
    \label{thm:tight_factor_of_3}
    For any $\eta>0$, there is an instance in which our algorithm for \textsc{Minimum Sum of Radii} computes a solution whose cost is by a
    factor of at least $3-\eta$ larger than the objective value of the constructed dual solution.
\end{thm}
\begin{proof}
    Choose $h \in \N_{\geq 3}$ such that $\frac{6h-3}{h^2+h-1} \leq \eta$ and consider the instance $(X^h_k,k)$.
    Starting from $\alpha_j = 0 = \lambda$ for all $j \in X^h_k$, the algorithm raises these variables until one of the pairs with radius $1,2$, or $3$ becomes tight.
    Note that a pair $(i,1)$, $i \in X^h_k$, covers at most $h^2-h+1$ points, while a pair $(i,2)$ covers at most $h^2+h-1$ points and a pair $(i,3)$ always covers $h^2+h$ points.
    Observing that for $h \geq 3$ we have
    \[
        \frac{1}{h^2-h} < \frac{2}{h^2+h-2} < \frac{3}{h^2+h-1},
    \]
    it follows from the formula given for $\delta$ that the algorithm raises all dual variables only by~$\frac{1}{h^2-h}$.
    Due to this raise, exactly those pairs $(i,1)$ get tight for which $i \in \{v_1, \dots, v_h\}$ for one of the $k$ copies contained in $G^h_k$.
    Call the set of these tight pairs $\tilde{\mathcal{B}}$ and note that $|\comp(\tilde{\mathcal{B}})|=k$.
    Hence, at this point the algorithm stops, implying that in the constructed dual solution we have $\alpha_j = \frac{1}{h^2-h}$ for all $j \in X^h_k$ as well as $\lambda=\frac{1}{h^2-h}$.
    This gives a total value for the dual solution $(\lambda, \alpha)$ of
    \[
        \sum_{j \in X^h_k} \alpha_j- k \cdot \lambda = k \cdot (h + h^2) \cdot \frac{1}{h^2-h} - k \cdot \frac{1}{h^2-h} = k \cdot  \frac{h^2+h-1}{h^2-h}.
    \]
    Our algorithm for \textsc{Minimum Sum of Radii} computes a feasible primal solution with cost at least $\OPT$, which by Lemma~\ref{lem: opt of X^h_k instance} is equal to $3k$.
    Thus, the cost of the computed primal solution will be larger than the objective value of the constructed dual solution by a factor of at least
    \[
        \displaystyle \frac{3k}{k \cdot \frac{h^2+h-1}{h^2-h}} = 3 \cdot \frac{h^2-h}{h^2+h-1} = 3 - \frac{6h-3}{h^2+h-1} \geq 3 - \eta.
    \]
    This completes the proof.
\end{proof}

The previously known primal-dual algorithms for \textsc{Minimum Sum of Radii} \cite{Ahmadian2016,Charikar2004} have a similar problem on this family of instances.
They are based on trying several values for $\lambda$ within a binary search framework.
For each fixed value of $\lambda$, they start with a solution in which $\alpha_{j}=0$ for each $j\in X'$.
Then, they raise $\alpha_{j}$ for each point $j\in X'$ simultaneously until some pair $(i,r) \in \calb$ becomes tight.
Finally, their binary search framework produces two values $\lambda_{1},\lambda_{2}$, corresponding constructed dual solutions, and sets of pairs $\calb_{1},\calb_{2}$.

In our family of instances above, it could be that, for a very small $\epsilon>0$, $\lambda_{1}=\frac{1}{h^2 - h} - \epsilon$ and $\lambda_{2} = \frac{1}{h^2 - h} + \epsilon$.
Then, the dual solution constructed for $\lambda_{1}$ is the same as the dual solution constructed by our algorithm above, the dual solution for $\lambda_{2}$ is essentially identical, and both solutions $\calb_{1},\calb_{2}$ contain the same set of $k$ pairs that our algorithm selected. Then, the computed set would be simply $\calb_{1}$ which is hence by a factor of $3$ more expensive than the corresponding constructed dual solution.

We remark that the algorithm in~\cite{Friggstad2022} can be implemented such that it solves an LP directly, instead of using a primal-dual algorithm to compute a solution.
However, the approximation ratio of $3$ is still a natural limit for this approach, since the algorithm intuitively selects some of the tight pairs, and each selected pair increases its radius by a factor of~$3$.

\section{$(3.5 + \epsilon)$-Approximation for the Generalized Lower Bounds Setting.}
\label{apx:generalized-lb}

In this section, we present our results for the \textsc{Minimum Sum of Radii} problem with generalized lower bounds.
As in the \textsc{Sum of Radii} problem, we are given a set of points $X$ in a metric space and an integer~$k$.
For any pair of points $i,j \in X$, we denote by $d(i,j)$ their distance. In addition, for each point $i \in X$ there is a set $\calx_i \subseteq 2^X$ of \emph{allowed sets of clients} for $i$, which is possibly given implicitly.
Each such set $\calx_i$ has the property that if $X' \in \calx_i$ then every superset of $X'$ is also in $\calx_i$.
A feasible solution is composed by a \emph{center set} $S \subseteq X$ of at most $k$ points, and an assignment $\sigma:X \rightarrow S$ that assigns each point $j$ to a point $\sigma(j) \in S$ such that $\sigma^{-1}(i) \in \calx_i$.
Hence, the set $\calx_i$ formalizes that not all sets of clients $X'\subseteq X$ can be assigned to $i$ (if $i$ is a cluster center) but that certain lower bounds need to be fulfilled, e.g., for the total number of the assigned clients, for their total weight (if the clients are weighted) or other characteristics.
The objective is to find a feasible solution that minimizes the sum of the distance from each $i \in S$ to the farthest point assigned to $i$, i.e., a tuple $(S,\sigma)$ that minimizes $c(S,\sigma) := \sum_{i \in S} \max_{j: \sigma(j) = i} d(i,j)$.

In the setting with outliers, we are additionally given an integer $m$ and we are allowed to assign up to $m$ points as outliers, i.e., the assignment is a function $\sigma: X \rightarrow S \cup \{\out\}$ for which some point $j$ can have assignment $\sigma(j) = \out$ but we require that $|\sigma^{-1}(\out)| \leq m$.
Notice that the usual \textsc{Minimum Sum of Radii} problem is the special case of the generalized setting where for each point $i$ the set
$\calx_i$ contains all subsets of $X$.

In the classical setting without lower bounds, in a feasible solution, a point $j \in X$ may be covered by two distinct centers $i$ and $i'$ with radii $r$ and $r'$, respectively, i.e., $j \in B(i,r)$ and also $j \in B(i',r')$.
In the setting with lower bounds, however, one needs to decide to which of these centers $j$ is assigned.
It might well be that all points are covered by the selected centers and radii, but there is no assignment of the points in $X$ to the centers such that $\sigma^{-1}(i) \in \calx_i$ and $\sigma^{-1}(i') \in \calx_{i'}$.

Recall that in the classical setting, we considered pairs $(i,r)$ which correspond to balls $B(i,r) := \{j \in X: d(i,j) \leq r\}$.
We guessed $1/\epsilon$ pairs of the optimal solution to construct a set of pairs $\calb$ that contains all the remaining pairs of $\OPT$.
In particular, we discarded pairs whose respective radius was too large.
For the generalized lower bound setting, we also consider a similar set, but we further remove pairs whose covered points do not satisfy the set of allowed clients.
For each possible center $i\in X$ let $d_i$ denote the smallest radius such that $B(i,d_i) \in \calx_i$.
For our algorithm, we do not need to know the sets $\calx_i$ explicitly.
It suffices to know $d_i$ for each point $i\in X$. Note that we could compute $d_i$ in polynomial time if we had oracle-access to $\calx_i$, i.e., if there were an oracle that we could query whether a set $X' \subseteq X$ is contained in $\calx_i$.
However, in the following, we will assume only that we know $d_i$ for each $i\in X$, without any further information about~$\calx_i$.

We define the set $\calb := \{(i,r) : i \in X \wedge \exists j \in X' ~s.t.~ d_i \leq r = d(i,j) \leq r_{1/\epsilon} \wedge B(i,r) \in \calx_i\}$ which is hence the set of pairs for which the radius is in the range $[d_i, r_{1/\epsilon}]$ and the radius also equals the distances between $i$ and some other point $j$ (w.l.o.g. we can restrict ourselves to such radii).
We would like to guess the $1/\epsilon$ pairs from $\OPT$ with largest radii.
However, in the setting with lower bounds, it does not suffice to guess those, as we would also need to guess which clients are assigned to them.
It is not clear how to do this in polynomial time.

Instead, we argue that if we guess a set of $1/\epsilon$ pairs $\calp_g$ from $\OPT$ and compute a $\alpha$-approximate solution
$\calp$ for the remaining problem, we can combine them to a $\max\{\alpha, 2\}$-approximate solution.
In this context, we say that two pairs $(i,r)$ and $(i'r')$ are disjoint if their respective balls are disjoint, i.e., $B(i,r) \cap B(i',r') = \emptyset$.

\begin{lem}[Ahmadian and Swamy~\cite{Ahmadian2016} - rephrased]
    \label{lem:reduction-glb}
    Let $\calp_g$ be a set of pairs in $\OPT$ and let $\calp = \{(i_1,r_1),\dots,(i_h,r_h)\}$ be a subset of pairs of $\calb$.
    Suppose that for each pair $(i_\ell,r_\ell) \in \calp$ there is a $(i_\ell,r'_\ell) \in \calb$ with $r'_\ell \leq r_\ell$ such that the pairs in the set $\{(i_\ell,r'_\ell) : (i_\ell,r_\ell) \in \calp \wedge \ell \in [h]\}$ are pairwise disjoint.
    Then, a solution $(S,\sigma)$ to the \textsc{Minimum Sum of Radii} problem with generalized lower bounds with cost $c(S,\sigma) \leq \sr(\calp) + 2 \cdot \sr(\calp_g)$ can be obtained in polynomial~time.
\end{lem}
\begin{proof}
    We describe an assignment procedure that yields the result.
    We choose the center set $S$ to include every point that had at least one point assigned to it.

    Points not covered by $\calp_g \cup \calp$ are assigned to $\out$.
    Since $\calp$ has at most $m$ points not covered, we have that~$|\sigma^{-1}(\out)| \leq m$.

    For every $\ell \in [h]$, assign to $i_\ell$ every point $j \in B(i_\ell,r'_\ell)$.
    Since the corresponding balls of pairs in $\{(i_\ell,r'_\ell) : (i_\ell,r_\ell) \in \calp \wedge \ell \in [h]\}$ are pairwise disjoint, no point is assigned to distinct centers, and since $B(i_\ell,r'_\ell) \in \calx_i$ for each $\ell$, we have that the set of points assigned to each $i_\ell$ is still an allowed set of clients even if further points are assigned to $i_\ell$.
    In this process, each center incurs a cost of $r'_\ell$, totalizing an assignment cost of $\sum_{\ell = 1}^h r'_\ell$.

    Initialize $\calp^*_g := \calp_g$.
    Iteratively, for every $\ell \in [h]$, let $C_\ell$ denote the component containing $(i_\ell, r_\ell)$ in $\calp^*_g \cup \{(i_\ell, r_\ell)\}$, assign to $i_\ell$ every point covered by $C_\ell$ that is not yet assigned to any point, and remove from $\calp^*_g$ every pair covering a point that was assigned.
    In this process, each $i_\ell$ that had a point assigned to it changes its incurring cost to $r_\ell + 2 \cdot \sr(C_\ell \setminus \{(i_\ell, r_\ell)\})$, and since the set $\bigcup_{i=1}^h \{C_\ell\}$ is pairwise disjoint, the total cost of the assignment is now at most $\sum_{\ell = 1}^h (r_\ell + 2\cdot \sr(C_\ell \setminus \{(i_\ell, r_\ell)\}))$.

    Let $\calp^*_g = \{(i^*_1,r^*_1), \dots, (i^*_{h^*},r^*_{h^*})\}$, denote an remaining pairs in $\calp^*_g$.
    For each pair, $(i^*_\ell,r^*_\ell) \in \calp^*_g$, let $C^*_\ell$ denote the component containing $(i^*_\ell, r^*_\ell)$ in $\calp^*_g$, and assign to $i^*_\ell$ every point covered by $C^*_\ell$ (these points have not yet been assigned to a center yet).
    In this process, each considered center $i^*_\ell$ incurs a cost of at most $2 \cdot \sr(C^*_\ell)$, and since the set $\bigcup_{i=1}^{h^*} \{C^*_\ell\}$ is pairwise disjoint, the total cost of the assignment is now at most $\sum_{\ell = 1}^{h} (r_\ell + 2 \cdot \sr(C_\ell \setminus \{(i_\ell, r_\ell)\})) + \sum_{\ell = 1}^{h^*} 2 \cdot \sr(C^*_\ell)$.

    After the last process, we have assigned every point in $X$ to some center $S$ or to $\out$, and since $|S| \leq |\calp \cup \calp_g| \leq k$ and $|\sigma^{-1}(\out)| \leq m$, the constructed solution $(S,\sigma)$ is feasible and it costs $c(S,\sigma) \leq \sr(\calp) + 2 \cdot \sr(\calp_g)$.
\end{proof}

Like before, we guess the $1/\epsilon$ pairs in $\OPT$ with largest radii.
We ignore the points covered by them for the remaining problem. However, note that these points are still relevant for our definition of $\calb$ above.
Then, our algorithm for the generalized setting is essentially the same as for the classical setting one (for both cases: with and without outliers), but with the modification that when replacing the pairs of a component $C$ for a single pair, we only consider pairs whose center equals the center of some pair in $C$.
For each $(i,r)$ in our final solution, let $(i,r')$ denote the corresponding pair as in Lemma~\ref{lem:reduction-glb} in the component $C$ that $(i,r)$ covers.
Note that because $r \geq r'$, it follows that $(i,r)$ covers $X(C)$ and, in particular, it also covers the points that $(i,r')$ covers.
Hence $B(i,r) \supseteq B(i,r') \in \calx_i$, and therefore $B(i,r) \in \calx_i$ is an allowed client set for $i$.
Also, since each of the pairs $(i,r')$ is contained in distinct components, the set of such pairs $(i,r')$ is pairwise disjoint, thus we may apply~Lemma~\ref{lem:reduction-glb}.

Although the algorithm remains fundamentally the same, the bound for the radius of each pair that replaces a component in the final solution changes.
In Lemma~\ref{lem:normal-component}, for the classical setting, we obtain a factor of $3$, using that \emph{any} point covered by a component can be the center of a pair in the final solution.
Since this is not the case for the generalized setting, we obtain only a factor of $3.5$, as shown in Lemma~\ref{lem:normal-component-glb}.
We will prove Lemma~\ref{lem:normal-component-glb} in Appendix~\ref{apx:lemma7} (it uses auxiliary definitions and results that are also used in the proof of Lemma~\ref{lem:normal-component} which can be found also in Appendix~\ref{apx:lemma7}).

\begin{restatable}{lem}{lemnormalcomponentglb}
    \label{lem:normal-component-glb}
    Let $C$ be a component of a set of tight pairs $\calb'$.
    Then there is a pair $(i,r)$ with $i\in X$ and $r \in \mathbb{R}$ and a set $C_d \subseteq C$ whose corresponding balls are pairwise disjoint such that
    \begin{enumerate}[label=$\arabic*.$,ref={\mbox{\rm $\arabic*$}}, noitemsep, leftmargin=1.5cm]
        \item there is a pair $(i, r') \in C$ with $r' \leq r$,
        \item $X(C)\subseteq B(i,r)$, and
        \item $r \leq 3.5 \cdot \sr(C_d) \le 3.5 \cdot \left(\sum_{j\in X'(C)}\alpha_{j}-\lambda\right) + \frac{3.5 \epsilon}{|X|}\OPT$.
    \end{enumerate}
\end{restatable}

This yields Theorem~\ref{thm:3+eps-glb}.

\begin{thm}
    \label{thm:3+eps-glb}
    For any $\epsilon > 0$ there is a $(3.5 + \epsilon)$-approximation for the \textsc{Minimum Sum of Radii} with generalized lower bounds with (and without) outliers.
\end{thm}

\section{Proof of Lemmas~\ref{lem:normal-component} and~\ref{lem:normal-component-glb}}
\label{apx:lemma7}

In this section, we work towards proving the main bound used to obtain the factors of $3$ for the classic setting, and $3.5$ for the setting with generalized lower bounds.
For this we need some auxiliary definitions and lemmas.

\begin{restatable}{defn}{defncomponentgraph}
    \label{defn:component-graph}
    Let $C$ be a component of $\calb' \subseteq \calb$, and let $C'$ be a maximal subset of such that if $(i,r),(i',r') \in C'$ then $B(i,r) \nsubseteq B(i',r')$.
    \begin{enumerate}[label=$\arabic*.$,ref={\mbox{\rm $\arabic*$}}, noitemsep, leftmargin=1.5cm]
        \item $V$ has a \emph{point vertex} $p_j$ for each point $j \in X(C)$, and a \emph{pair vertex} $b_{(i,r)}$ and a \emph{frontier vertex} $t_{(i,r)}$ for each pair $(i,r) \in C'$,
        \item $\{b_{(i,r)}, p_j\} \in E$ and has weight $w(b_{(i,r)}, p_j) = r$ if $p_j \in B(i,r)$,
        \item $\{b_{(i,r)}, t_{(i,r)}\} \in E$ and has weight $w(b_{(i,r)}, t_{(i,r)}) = r$ for every $(i,r) \in C$.
    \end{enumerate}
\end{restatable}
Let $P$ be a path in $G$, then we define its weight by $w(P) := \sum_{e \in E(P)} w(e)$, and we say that a path $P$ from vertex $u$ to $v$ in $G$ is a \emph{shortest path} if it minimizes $w(P)$ among all paths from $u$ to $v$.
Let $P(u,v)$ denote a shortest path from vertex $u$ to vertex $v$ in $G$, then we define the \emph{eccentricity} $\varepsilon(u) := \max_{v \in V(G)} w(P(u,v))$ of $u$ as the weight of the heaviest shortest path that starts on $u$.
Then, we define the \emph{radius} $\rad(G) := \min_{u \in V(G)} \varepsilon(u)$ of $G$ as the weight of the shortest longest shortest path between pairs of vertices in $G$, and we say that a vertex $c$ is a \emph{center} of $G$ if $\varepsilon(u) = \min_{u \in V(G)} \varepsilon(u)$, i.e., the weight of the distance from $u$ to any other vertex in $G$ is at most $\rad(G)$.
Notice that a frontier vertex $t_{(i,r)}$ can never be a center of $G$, unless $r=0$, because it is always further from every other vertex than $b_{(i,r)}$ is, and because point vertices are also only adjacent to pair vertices covering them, every path on $G$ always alternates between non-pair vertices and pair vertices.
In this context, we are interested in the following kind of paths on~$G$.
\begin{defn}
    In a component graph $G$, we say that a path $P = (v_0,\dots,v_h)$ is \emph{nice} if
    \begin{enumerate}[label=$\arabic*.$,ref={\mbox{\rm $\arabic*$}}, noitemsep, leftmargin=1.5cm]
        \item $v_0$ and $v_h$ are not pair vertices,
        \item The respective balls of pair vertices $v_i$ and $v_j$ have an empty intersection for every $i + 2 < j$.
    \end{enumerate}
\end{defn}

\begin{lem}
    \label{lem:nice-shortest-path}
    Let $G$ be the component graph of a component of $\calb' \subseteq \calb$, and $u,v \in V(G)$ be two non-pair vertices.
    Then, among the shortest paths from $v$ to $u$ in $G$, there is always one that is nice.
\end{lem}
\begin{proof}
    Let $P = (v_0 = v, \dots, v_n = u)$ be a non-nice shortest path starting and ending on $v$ and $u$.
    Then, there are $i,j$ with $i + 2 < j$ such that their respective balls $B_i, B_j$ have a non-empty intersection.
    But this implies that there is a path $P' = (v_0,\dots,v_i,p,v_j,\dots,v_n)$, where $p \in B_i \cap B_j$, such that
    \begin{align*}
        w(P)
         & \leq w(P')
        = w((v_0,\dots,v_i)) + w(\{v_i,p\}) + w(\{p,v_j\}) + w((p,v_j,\dots,v_n))                \\
         & = w((v_0,\dots,v_i)) + w(\{v_i,v_{i+1}\}) + w(\{v_{j-1},v_j\}) + w((p,v_j,\dots,v_n)) \\
         & \leq w(P),
    \end{align*}
    and thus $P'$ is also a shortest path.
    Because the number of pair vertices in the path is reduced by at least one each time this shortcutting procedure is applied, iteratively applying it then leads to a nice shortest path from $v$~to~$u$.
\end{proof}

For a path $P$, let $b(P)$ denote the subset of pair vertices in $P$ and let $b_d(P)$ denote a subset of pair vertices in $b(P)$ whose respective pairs are pairwise disjoint and maximizes $\sum_{b_{(i',r')} \in b_d(P)} r'$.
Observe that the definition of $b_d(P)$ implies that $\sum_{b_{(i',r')} \in b(P)} r' \leq 2 \sum_{b_{(i',r')} \in b_d(P)} r'$.
We are interested in nice paths because, for any nice path $P = (v_0,\dots,v_h)$, the sum of the radii of the pairs whose respective pair vertices are in $b(P)$ is exactly half of $w(P)$, and since the balls corresponding to pair vertices $v_i,v_j \in P$ have an empty intersection for every $i + 2 < j$, it follows that the sum of the radii of the pairs whose respective pair vertices are in $b_d(P)$ is then at least a quarter of $w(P)$, i.e.
\begin{align}
    \label{coro:normal-component-ineq}
    w(P)
     & = 2 \sum_{b_{(i',r')} \in b(P)} r'
    \leq 4 \sum_{b_{(i',r')} \in b_d(P)} r'
    \leq 4 \sum_{(i',r') \in C_d} r'
    = 4 \cdot \sr(C_d),
\end{align}
where $C_d$ denotes a subset of $C$ whose corresponding balls are pairwise disjoint and maximizes~$\sr(C_d)$, hence $\sum_{b_{(i',r')} \in b_d(P)} r' \leq \sum_{(i',r') \in C_d} r'$.

We say that a set $\calp = P_1,\dots,P_h$ of paths is \emph{nice} if every path in $\calp$ is nice and no pair vertex in $P_i$ has a common adjacent point vertex with a pair vertex in $P_j$ for every $i\neq j$.
Also, we define the weight $w(\calp) := \sum_{P \in \calp} w(P)$.

\begin{lem}
    \label{lem:nice-path-collection}
    Let $P$ be a nice path in the component graph $G$ of component $C$ and $(\overline{i},\overline{r})$ be the pair with the largest radius in $C$.
    Then there is a nice set of paths $\calp$ such that
    \begin{align*}
        4 \cdot \sr(C_d)
         & \geq w(\calp)
        \geq \rad(G) + \min\left\{ \frac{1}{2} w(P) - \frac{1}{2} \overline{r}, w(P) - 3 \overline{r} \right\}.
    \end{align*}
\end{lem}
\begin{proof}
    For any $u,v \in V$, let $P(u,v)$ denote a shortest path from $u$ to $v$, and assume, w.l.o.g., that for any shortest path $P(u,v)$ and $u',v' \in P(u,v)$ we have $P(u',v') \subseteq P(u,v)$.
    By Lemma~\ref{lem:nice-shortest-path}, we may assume that $P(u,v)$ is nice and so is every subpath of $P(u,v)$.

    Let $P = P(v,x)$,  let $x'$ be the vertex closest to the ``middle'' of the path $P(v,x)$, i.e., $x' \in P(v,x)$ is a vertex that minimizes $\left|\frac{1}{2} w(P(v,x) - w(P(v,x')) \right|$, and let $r_{x'}$ denote the weight of the heaviest edges incident to $x'$ in this path.
    Then,
    \begin{align}
        w(P(v,x'))
         & \geq \frac{1}{2} w(P(v,x)) - \frac{1}{2} r_{x'}, \text{ and }
        w(P(x',x))
        \geq \frac{1}{2} w(P(v,x)) - \frac{1}{2} r_{x'}.
        \label{lem:nice-path-collection:ineq-1}
    \end{align}

    Let $y$ be a farthest vertex from $x'$ in $G$ ($y$ is either a frontier vertex or a point vertex), and let $y'$ denote the first vertex in the path $P(y,x')$ that is adjacent to a vertex in $P(v,x)$.
    Since $y$ is a farthest vertex from $x'$, we have that $P(y,x') \geq \rad(G)$.
    Next, we analyze three distinct cases: (i) $y'$ is not adjacent to any vertex in $P(x',x)$; (ii) $y'$ is not adjacent to any vertex in $P(v,x')$; and (iii) $y'$ is adjacent to a vertex in $P(v,x')$ and to a vertex in $P(x',x)$.
    The three cases are depicted in Figure~\ref{fig: construction of paths}.

    \begin{figure}[h!]
        \centering
        \begin{minipage}[t]{.3\textwidth} 
            \centering
            \begin{tikzpicture}[scale=0.9]

                \path[-,color=cyan!50]
                (0.5,1.5) edge[line width=3pt] (2.4,0.3)
                (2.5,-0.4) edge[line width=3pt] (2.4,0.3)
                (2.5,-0.4) edge[line width=3pt] (3.1,-0.6)
                (3.3,-1.2) edge[line width=3pt] (3.1,-0.6)
                (3.3,-1.2) edge[line width=3pt] (2.9,-1.7)
                (0.5,1.5) edge[line width=3pt] (1,2.2)
                (1.8,2) edge[line width=3pt] (1,2.2)
                (1.8,2) edge[line width=3pt] (2.6,2.3)
                (2.8,2.9) edge[line width=3pt] (2.6,2.3)
                (2.8,2.9) edge[line width=3pt] (3.5,3)
                (4.3,3.6) edge[line width=3pt] (3.5,3)
                (4.3,3.6) edge[line width=3pt] (3.5,4.2)
                ;

                \path[-]
                (0,0) edge[line width=0.8pt] (0.5,0.7)
                (0.5,1.5) edge[line width=0.8pt] (0.5,0.7)
                (0.5,1.5) edge[line width=0.8pt] (1,2.2)
                (1.8,2) edge[line width=0.8pt] (1,2.2)
                (1.8,2) edge[line width=0.8pt] (2.6,2.3)
                (2.8,2.9) edge[line width=0.8pt] (2.6,2.3)
                (2.8,2.9) edge[line width=0.8pt] (3.5,3)
                (4.3,3.6) edge[line width=0.8pt] (3.5,3)
                (4.3,3.6) edge[line width=0.8pt] (3.5,4.2)
                ;

                \path[-]
                (1.8,2) edge[line width=0.8pt] (1.8,1.3)
                (2.5,1) edge[line width=0.8pt] (1.8,1.3)
                (2.5,1) edge[line width=0.8pt] (2.4,0.3)
                (2.5,-0.4) edge[line width=0.8pt] (2.4,0.3)
                (2.5,-0.4) edge[line width=0.8pt] (3.1,-0.6)
                (3.3,-1.2) edge[line width=0.8pt] (3.1,-0.6)
                (3.3,-1.2) edge[line width=0.8pt] (2.9,-1.7)
                ;

                \path[dashed]
                (0.5,0.7) edge[line width=0.8pt] (2.4,0.3)
                (0.5,1.5) edge[line width=0.8pt] (2.4,0.3)
                ;

                \fill (0,0) circle (3pt) node[left,color=black] {$v$};
                \fill (1.8,2) circle (3pt) node[above,color=black] {$x'$};
                \fill (3.5,4.2) circle (3pt) node[above,color=black] {$x$};
                \fill (2.4,0.3) circle (3pt) node[right,color=black] {$y'$};
                \fill (2.9,-1.7) circle (3pt) node[left,color=black] {$y$};
                \fill (0.5,1.5) circle (3pt) node[left,color=black] {$x''$};

            \end{tikzpicture}
            \subcaption*{\textbf{Case 1:} $y'$ not adjacent to any vertex in $P(x',x)$.}
        \end{minipage}
        \begin{minipage}[t]{.3\textwidth}
            \centering
            \begin{tikzpicture}[scale=0.9]

                \path[-,color=cyan!50]
                (0,0) edge[line width=3pt] (0.5,0.7)
                (0.5,1.5) edge[line width=3pt] (0.5,0.7)
                (0.5,1.5) edge[line width=3pt] (1,2.2)
                (1.8,2) edge[line width=3pt] (1,2.2)
                (1.8,2) edge[line width=3pt] (2.6,2.3)
                (2.8,2.9) edge[line width=3pt] (2.6,2.3)
                (2.8,2.9) edge[line width=3pt] (3.5,3)
                (3.5,3) edge[line width=3pt] (2.4,0.3)
                (2.5,-0.4) edge[line width=3pt] (2.4,0.3)
                (2.5,-0.4) edge[line width=3pt] (3.1,-0.6)
                (3.3,-1.2) edge[line width=3pt] (3.1,-0.6)
                (3.3,-1.2) edge[line width=3pt] (2.9,-1.7)
                ;

                \path[-]
                (0,0) edge[line width=0.8pt] (0.5,0.7)
                (0.5,1.5) edge[line width=0.8pt] (0.5,0.7)
                (0.5,1.5) edge[line width=0.8pt] (1,2.2)
                (1.8,2) edge[line width=0.8pt] (1,2.2)
                (1.8,2) edge[line width=0.8pt] (2.6,2.3)
                (2.8,2.9) edge[line width=0.8pt] (2.6,2.3)
                (2.8,2.9) edge[line width=0.8pt] (3.5,3)
                (4.3,3.6) edge[line width=0.8pt] (3.5,3)
                (4.3,3.6) edge[line width=0.8pt] (3.5,4.2)
                ;

                \path[-]
                (1.8,2) edge[line width=0.8pt] (1.8,1.3)
                (2.5,1) edge[line width=0.8pt] (1.8,1.3)
                (2.5,1) edge[line width=0.8pt] (2.4,0.3)
                (2.5,-0.4) edge[line width=0.8pt] (2.4,0.3)
                (2.5,-0.4) edge[line width=0.8pt] (3.1,-0.6)
                (3.3,-1.2) edge[line width=0.8pt] (3.1,-0.6)
                (3.3,-1.2) edge[line width=0.8pt] (2.9,-1.7)
                ;

                \path[dashed]
                (3.5,3) edge[line width=0.8pt] (2.4,0.3)
                (4.3,3.6) edge[line width=0.8pt] (2.4,0.3)
                ;

                \fill (0,0) circle (3pt) node[left,color=black] {$v$};
                \fill (1.8,2) circle (3pt) node[above,color=black] {$x'$};
                \fill (3.5,4.2) circle (3pt) node[above,color=black] {$x$};
                \fill (2.4,0.3) circle (3pt) node[right,color=black] {$y'$};
                \fill (2.9,-1.7) circle (3pt) node[left,color=black] {$y$};
                \fill (3.5,3) circle (3pt) node[above,color=black] {$x''$};

            \end{tikzpicture}
            \subcaption*{\textbf{Case 2:} $y'$ not adjacent to any vertex in $P(v,x')$.}
        \end{minipage}
        \begin{minipage}[t]{.3\textwidth} 
            \centering
            \begin{tikzpicture}[scale=0.9]

                \path[-,color=cyan!50]
                (0,0) edge[line width=3pt] (0.5,0.7)
                (0.5,1.5) edge[line width=3pt] (0.5,0.7)
                (0.5,1.5) edge[line width=3pt] (1,2.2)
                (1.8,2) edge[line width=3pt] (1,2.2)
                (1.8,2) edge[line width=3pt] (2.6,2.3)
                (2.8,2.9) edge[line width=3pt] (2.6,2.3)
                (2.8,2.9) edge[line width=3pt] (3.5,3)
                (4.3,3.6) edge[line width=3pt] (3.5,3)
                (4.3,3.6) edge[line width=3pt] (3.5,4.2)
                (2.5,-0.4) edge[line width=3pt] (3.1,-0.6)
                (3.3,-1.2) edge[line width=3pt] (3.1,-0.6)
                (3.3,-1.2) edge[line width=3pt] (2.9,-1.7)
                ;

                \path[-]
                (0,0) edge[line width=0.8pt] (0.5,0.7)
                (0.5,1.5) edge[line width=0.8pt] (0.5,0.7)
                (0.5,1.5) edge[line width=0.8pt] (1,2.2)
                (1.8,2) edge[line width=0.8pt] (1,2.2)
                (1.8,2) edge[line width=0.8pt] (2.6,2.3)
                (2.8,2.9) edge[line width=0.8pt] (2.6,2.3)
                (2.8,2.9) edge[line width=0.8pt] (3.5,3)
                (4.3,3.6) edge[line width=0.8pt] (3.5,3)
                (4.3,3.6) edge[line width=0.8pt] (3.5,4.2)
                ;

                \path[-]
                (1.8,2) edge[line width=0.8pt] (1.8,1.3)
                (2.5,1) edge[line width=0.8pt] (1.8,1.3)
                (2.5,1) edge[line width=0.8pt] (2.4,0.3)
                (2.5,-0.4) edge[line width=0.8pt] (2.4,0.3)
                (2.5,-0.4) edge[line width=0.8pt] (3.1,-0.6)
                (3.3,-1.2) edge[line width=0.8pt] (3.1,-0.6)
                (3.3,-1.2) edge[line width=0.8pt] (2.9,-1.7)
                ;

                \path[dashed]
                (0.5,0.7) edge[line width=0.8pt] (2.4,0.3)
                (0.5,1.5) edge[line width=0.8pt] (2.4,0.3)
                (3.5,3) edge[line width=0.8pt] (2.4,0.3)
                (4.3,3.6) edge[line width=0.8pt] (2.4,0.3)
                ;

                \fill (0,0) circle (3pt) node[left,color=black] {$v$};
                \fill (1.8,2) circle (3pt) node[above,color=black] {$x'$};
                \fill (3.5,4.2) circle (3pt) node[above,color=black] {$x$};
                \fill (2.4,0.3) circle (3pt) node[right,color=black] {$y'$};
                \fill (2.9,-1.7) circle (3pt) node[left,color=black] {$y$};
                \fill (0.5,1.5) circle (3pt) node[left,color=black] {$x_1$};
                \fill (3.5,3) circle (3pt) node[above,color=black] {$x_2$};
                \fill (2.5,-0.4) circle (3pt) node[left,color=black] {$y''$};

            \end{tikzpicture}
            \subcaption*{\textbf{Case 3:} $y'$ adjacent to vertex in $P(v,x')$ and $P(x',x)$.}
        \end{minipage}
        \caption{Construction of a set of nice paths as described in Lemma \ref{lem:nice-path-collection}.}
        \label{fig: construction of paths}
    \end{figure}

    \textbf{Case 1:} $y'$ is not adjacent to any vertex in $P(x', x)$.
    Let $x''$ be a vertex in $P(v,x')$ closest to $x'$ that is adjacent to $y'$.
    Let $P(y,y') = (y_1 = y, \cdots, y_q = y')$ and $P(x'',x) = (x_1 = x'', \cdots, x_p = x)$, and consider the nice path $P = (y_1 = y, \cdots, y_q = y', x_1 = x'', \cdots, x_p = x)$.
    Then, by the triangular inequality
    \begin{align*}
        w(P)
         & \stackrel{\phantom{\scriptsize \mathrm{(ineq.\ref{lem:nice-path-collection:ineq-1})}}}{=} w(P(y,y')) + w(y',x'') + w(P(x'',x))                      \\
         & \stackrel{\phantom{\scriptsize \mathrm{(ineq.\ref{lem:nice-path-collection:ineq-1})}}}{=} w(P(y,y')) + w(y',x'') + w(P(x'',x')) + w(P(x',x))        \\
         & \stackrel{\phantom{\scriptsize \mathrm{(ineq.\ref{lem:nice-path-collection:ineq-1})}}}{\geq} w(P(y,x')) + w(P(x',x))                                \\
         & \stackrel{\scriptsize \mathrm{(ineq.\ref{lem:nice-path-collection:ineq-1})}}{\geq} \rad(G) + \frac{1}{2} w(P(v,x')) - \frac{1}{2} r_{x'}            \\
         & \stackrel{\phantom{\scriptsize \mathrm{(ineq.\ref{lem:nice-path-collection:ineq-1})}}}{\geq} \rad(G) + \frac{1}{2} w(P) - \frac{1}{2} \overline{r}.
    \end{align*}

    \textbf{Case 2:} $y'$ is not adjacent to any vertex in $P(v, x')$.
    This case is analogous to the previous one.
    Let $x''$ be a vertex in $P(x',x)$ closest to $x'$ that is adjacent to $y'$.
    Let $P(x'',v) = (x_1 = x'', \cdots, x_p = v)$, and consider the nice path $P = (y_1 = y, \cdots, y_q = y', x_1 = x'', \cdots, x_p = v)$.
    Then, by the triangular inequality
    \begin{align*}
        w(P)
         & \stackrel{\phantom{\scriptsize \mathrm{(ineq.\ref{lem:nice-path-collection:ineq-1})}}}{=} w(P(y,y')) + w(y',x'') + w(P(x'',v))                      \\
         & \stackrel{\phantom{\scriptsize \mathrm{(ineq.\ref{lem:nice-path-collection:ineq-1})}}}{=} w(P(y,y')) + w(y',x'') + w(P(x'',x')) + w(P(x',v))        \\
         & \stackrel{\phantom{\scriptsize \mathrm{(ineq.\ref{lem:nice-path-collection:ineq-1})}}}{\geq} w(P(y,x')) + w(P(x',v))                                \\
         & \stackrel{\scriptsize \mathrm{(ineq.\ref{lem:nice-path-collection:ineq-1})}}{\geq} \rad(G) + \frac{1}{2} w(P(v,x')) - \frac{1}{2} r_{x'}            \\
         & \stackrel{\phantom{\scriptsize \mathrm{(ineq.\ref{lem:nice-path-collection:ineq-1})}}}{\geq} \rad(G) + \frac{1}{2} w(P) - \frac{1}{2} \overline{r}.
    \end{align*}

    \textbf{Case 3:} $y'$ is adjacent to a vertex in $P(v, x')$ and to a vertex in $P(x',x)$.
    First, we show that we may assume w.l.o.g. that $y'$ is a pair vertex.
    Suppose that $y'$ is a point vertex.
    Then, there is a pair vertex $b_1 \in P(v,x')$ and $b_2 \in P(x',x)$ to which $y'$ is adjacent to.
    Since $P(v,x) = (x_0 = v,\dots,x_h = x)$ is nice, it must be that $x_i = b_1$ and $x_{i+2} = b_2$ for some $i$.
    And because from the construction of $G$ we have that $w(b_1,x_{i+1}) = w(b_1,y')$ and $w(x_{i+1},b_2) = w(y',b_2)$, we can replace point vertex $x_{i+1}$ for $y'$ and still maintain that that $P(v,x)$ is a nice shortest path.

    Now, assume that $y' = b_{(i',r')}$ is a pair vertex, and let $x_1 \in P(v,x')$ be a point vertex that is adjacent to $y'$ and is closest to $x'$ and let $x_2 \in P(x',x)$ be a point vertex that is adjacent to $y'$ and is closest to $x'$.
    Let $w_{x'}$ denote the shorter between $w(P(x_1,x'))$ and $w(P(x',x_2))$.
    Then, by the triangular inequality
    \begin{align}
        \label{lem:nice-path-collection:ineq-2}
        2 w_{x'}
         & \leq w(P(x_1,x')) + w(P(x',x_2))
        = w(P(x_1,x_2)) \nonumber           \\
         & \leq w(P(x_1,y')) + w(P(y',x_2))
        = w(x_1,y') + w(y',x_2)
        = 2 r'
        \leq 2 \overline{r}.
    \end{align}
    Let $y''$ be the point vertex that precedes $y'$ in $P(y,y')$, and consider the path $P(y, y'')$.
    Then,
    \begin{align}
        \label{lem:nice-path-collection:ineq-3}
        w(P(y,y''))
         & = w(P(y,x')) - w_{x'} - 2 r'
        \stackrel{\scriptsize \mathrm{(ineq.\ref{lem:nice-path-collection:ineq-2})}}{\geq} \rad(G) - 3 \overline{r}.
    \end{align}
    Notice that no point vertex in $G$ is adjacent simultaneously to a pair vertex in $P(v,x)$ and to a pair vertex in $P(y, y'')$, and also notice that both paths are nice.
    Then,
    \begin{align*}
        w(P(v,x)) + w(P(y, y''))
         & \stackrel{{\scriptsize \mathrm{(ineq.\ref{lem:nice-path-collection:ineq-3})}}}{\geq} w(P) + \rad(G) - 3 \overline{r}.
    \end{align*}

    Therefore, considering the three cases analyzed, we can always find a nice set $\calp$ whose weight is at least
    \begin{align*}
        w(\calp) \geq \rad(G) + \min\left\{  \frac{1}{2} w(P) - \frac{1}{2} \overline{r}, w(P) - 3 \overline{r} \right\},
    \end{align*}
    and because of inequality~(\ref{coro:normal-component-ineq}), we also have that $4 \cdot \sr(C_d) \geq w(\calp)$, completing the proof.
\end{proof}

We now have the necessary ingredients to prove Lemmas~\ref{lem:normal-component} and~\ref{lem:normal-component-glb}.

\lemnormalcomponent*
\begin{proof}
    Choose $r = \rad(G)$ and $i$ as the corresponding point of a center $c$ of $G$, i.e., $i = i'$ if $c = p_{i'}$ is a point vertex or $i = i'$ if $c = b_{(i',r')}$ is a pair vertex (recall that frontier vertices are never centers of $G$).
    Since $w(c,v) \geq d(i,j)$ for every vertex $v$ and its corresponding point $j \in X$, it follows that $X(C) \subseteq B(i,r)$.
    Recall that $C_d$ denotes a pairwise disjoint subset of $C$ that maximizes $\sr(C_d)$.
    Since the pairs in $C_d$ are pairwise disjoint and almost tight, we have that
    \begin{align*}
        \sr(C_d)
         & = \sum_{(i',r') \in C_d} r'
        \leq \left( \sum_{j\in X'(C)} \alpha_{j} - |C_d| \cdot \lambda + |C_d| \cdot \mu \right)
        \stackrel{{\scriptsize \mathrm{(Lem.\ref{lem:biggest-radius})}}}{\leq} \left( \sum_{j\in X'(C)} \alpha_{j} - \lambda \right) + \epsilon \frac{\OPT}{|X|}.
    \end{align*}
    Thus, to complete the proof it only remains to show that
    \begin{align}
        \label{lem:normal-component:ineq-1}
        \rad(G) \leq 3 \cdot \sr(C_d).
    \end{align}

    If there is a pair $(i',r') \in C$ such that $r' \geq \frac{1}{3} \rad(G)$, then choosing $C_d = \{(i',r')\}$ we have~(\ref{lem:normal-component:ineq-1}).
    Otherwise, we have that $r' < \frac{1}{3} \rad(G)$ for every pair $(i',r') \in C$.
    Let $(\overline{i},\overline{r})$ denote a pair with largest radius in $C$ and let $\overline{f}$ denote its respective frontier vertex.
    Let $x$ be a vertex in $G$ that is farthest from $b_{(\overline{i},\overline{r})}$ ($x$ is either a frontier vertex or a point vertex), then
    \begin{align}
        \label{lem:normal-component:ineq-2}
        w(P(\overline{f}, x))
         & = \overline{r} + w(P(b_{(\overline{i},\overline{r})}, x))
        \geq \overline{r} + \rad(G).
    \end{align}

    From Lemma~\ref{lem:nice-path-collection}, we have that
    \begin{align*}
        \rad(G)
         & \stackrel{{\scriptsize \mathrm{(Lem.\ref{lem:nice-path-collection})}}}{\leq} 4 \cdot \sr(C_d) - \min\left\{ \frac{1}{2} w(P(\overline{f}, x)) - \frac{1}{2} \overline{r}, w(P(\overline{f}, x)) - 3 \overline{r} \right\}                                                      \\
         & \stackrel{{\scriptsize \mathrm{(ineq.\ref{lem:normal-component:ineq-2})}}}{\mathmakebox[\widthof{$\stackrel{{\scriptsize \mathrm{(Lem.\ref{lem:nice-path-collection})}}}{\leq}$}]{\leq}} 4 \cdot \sr(C_d) - \min\left\{ \frac{1}{2} \rad(G), \rad(G) - 2 \overline{r} \right\} \\
         & \stackrel{\phantom{\scriptsize \mathrm{(ineq.\ref{lem:normal-component-glb:ineq-7})}}}{\leq} 4 \cdot \sr(C_d) - \min\left\{ \frac{1}{2} \rad(G), \rad(G) - 2 \cdot \frac{1}{3} \rad(G) \right\}                                                                                \\
         & \stackrel{\phantom{\scriptsize \mathrm{(ineq.\ref{lem:normal-component-glb:ineq-7})}}}{=} 4 \cdot \sr(C_d) - \frac{1}{3} \rad(G),
    \end{align*}
    which implies inequality~(\ref{lem:normal-component:ineq-1}), finishing the proof.
\end{proof}

From Appendix~\ref{apx:generalized-lb}, recall that in the setting with the generalized lower bounds we only replace the pairs in a component $C$ for a single pair $(i,r)$ if there is a pair $(i,r') \in C$, this is, there is a pair in $C$ whose center is colocated with the new pair $(i,r)$.
This restriction worsens our previous bound of $3$, stated in Lemma~\ref{lem:normal-component}, to $3.5$, as seen in Lemma~\ref{lem:normal-component-glb}.

\lemnormalcomponentglb*
\begin{proof}
    Consider the component graph $G$ of $C$.
    In this general setting, consider placing a pair centered on a point $i'$ only if there is pair vertex $b_{(i',r')}$ in $G$.
    This is because if $b_{(i',r')}$ is in $G$, then $(i',r') \in \calb$, which implies that any pair also centered in $i'$ that covers at least the same points as $(i',r')$ has radius at least $r'$.

    Let $w_{i'} := \max_{v \in V} w(P(b_{(i',r')},v))$ denote the longest distance from pair vertex $b_{(i',r')}$ to some vertex in $G$ (recall that from the definition of $G$, there is at most only one pair vertex for each $i'$).
    Then any pair $(i,r)$ centered on $i'$ for which $b_{(i',r')}$ is in $G$ covers $X(C)$ if $r \geq w_{i'}$ because $w_{i'} \geq d(i,j)$ for every point $j \in X$, hence $B(i',r') \subseteq B(i',w_{i'})$.

    Recall that $C_d$ denotes a pairwise disjoint subset of $C$ that maximizes $\sr(C_d)$.
    Since the pairs in $C_d$ are pairwise disjoint and almost tight, we have that
    \begin{align*}
        \sr(C_d)
         & = \sum_{(i',r') \in C_d} r'
        \leq \left( \sum_{j\in X'(C)} \alpha_{j} - |C_d| \cdot \lambda + |C_d| \cdot \mu \right)
        \stackrel{{\scriptsize \mathrm{(Lem.\ref{lem:biggest-radius})}}}{\leq} \left( \sum_{j\in X'(C)} \alpha_{j} - \lambda \right) + \epsilon \frac{\OPT}{|X|}.
    \end{align*}
    Then, to finish the proof, it remains only to show that there is a point $i'$, for which $b_{(i',r')}$ is in $G$, such that
    \begin{align}
        \label{lem:normal-component-glb:ineq-1}
        w_{i'} \leq \frac{7}{2} \cdot \sr(C_d).
    \end{align}
    We break the case into two main cases.

    \textbf{Case 1:} Assume that a center of $G$ is a pair vertex $b_{(i'',r'')}$.
    Then choosing $i = i''$ and $r = \rad(G) = w_{i''}$ yields a pair $(i,r)$ that covers $X(C)$.
    We apply the same reasonings as in Lemma~\ref{lem:normal-component}.

    \textbf{Case 1.1:} Assume that there is a pair $(i',r') \in C$ such that $r' \geq \frac{1}{3} \rad(G)$.
    Then
    \begin{align*}
        w_{i}
         & = \rad(G)
        \leq 3 r'
        \leq 3 \cdot \sr(C_d)
        \leq \frac{7}{2} \cdot \sr(C_d),
    \end{align*}
    and we have inequality~(\ref{lem:normal-component-glb:ineq-1}).

    \textbf{Case 1.2:} Assume that $r' < \frac{1}{3} \rad(G)$ for every pair $(i',r') \in C$.
    Let $(\overline{i},\overline{r})$ denote a pair with largest radius in $C$ and let $\overline{f}$ denote its respective frontier vertex.
    Let $x$ be a vertex in $G$ that is farthest from $b_{(\overline{i},\overline{r})}$ ($x$ is either a frontier vertex or a point vertex), then
    \begin{align}
        \label{lem:normal-component-glb:ineq-2}
        w(P(\overline{f}, x))
         & = \overline{r} + w(P(b_{(\overline{i},\overline{r})}, x))
        \geq \overline{r} + \rad(G).
    \end{align}
    From Lemma~\ref{lem:nice-path-collection}, we have that
    \begin{align*}
        w_{i}
         & \stackrel{\phantom{\scriptsize \mathrm{(ineq.\ref{lem:normal-component:ineq-2})}}}{=} \rad(G)
        \stackrel{{\scriptsize \mathrm{(Lem.\ref{lem:nice-path-collection})}}}{\leq} 4 \cdot \sr(C_d) - \min\left\{ \frac{1}{2} w(P(\overline{f}, x)) - \frac{1}{2} \overline{r}, w(P(\overline{f}, x)) - 3 \overline{r} \right\} \\
         & \stackrel{{\scriptsize \mathrm{(ineq.\ref{lem:normal-component:ineq-2})}}}{\leq} 4 \cdot \sr(C_d) - \min\left\{ \frac{1}{2} \rad(G), \rad(G) - 2 \overline{r} \right\}                                                 \\
         & \stackrel{\phantom{\scriptsize \mathrm{(ineq.\ref{lem:normal-component:ineq-2})}}}{\leq} 4 \cdot \sr(C_d) - \min\left\{ \frac{1}{2} \rad(G), \rad(G) - 2 \cdot \frac{1}{3} \rad(G) \right\}                            \\
         & \stackrel{\phantom{\scriptsize \mathrm{(ineq.\ref{lem:normal-component:ineq-2})}}}{=} 4 \cdot \sr(C_d) - \frac{1}{3} w_i,
    \end{align*}
    which implies that $w_{i} \leq 3 \cdot \sr(C_d) \leq \frac{7}{2} \cdot \sr(C_d)$, and we have inequality~(\ref{lem:normal-component-glb:ineq-1}).
    This proves the claim for Case~1.

    \textbf{Case 2:} Assume that every center of $G$ is a point vertex (recall that frontier vertices cannot be a center of $G$), and let $c$ denote one of these point vertices.
    We show that there is always a pair vertex that is not too far away from the center of $G$ and that the distance from this pair vertex to every other vertex in $G$ is not that much more than $\rad(G)$.
    Let $(i_1,r_1)$ be a pair in $C$ with smallest radius such that $b_{(i_1,r_1)}$ is adjacent to $c$.

    First, let us bound $w_{i_1}$.
    For some vertex $x$, if there is a point vertex $x' \neq c$ that is adjacent to a pair vertex in $P(c,x)$ and to $b_{(i_1,r_1)}$, then we show that the distance from $i_1$ to $x$ is relatively short.
    The first pair vertex in $P(c,x')$ has radius at least $r_1$, therefore $w(P(c,x')) \geq 2 r_1$, and by the triangular inequality, we have that
    \begin{align*}
        w(P(i_1,x))
         & \leq r_1 + w(P(x',x))
        = r_1 + w(P(c,x)) - w(P(c,x'))     \\
         & \leq r_1 + \rad(G) - w(P(c,x')) \\
         & \leq \rad(G) - r_1.
    \end{align*}
    Since $b_{(i_1,r_1)}$ is not a center of $G$, there must be a vertex that is farther than a distance $\rad(G)$ from it.
    Thus, there must be vertex $x$ such that there is no point vertex, besides $c$, that is adjacent to a pair vertex in $P(x,c)$ and to $b_{(i_1,r_1)}$ simultaneously, and $w(P(c,x)) > \rad(G) - r_1$.
    Thus, let $x$ be a point vertex be the farthest vertex from $c$ such that there is no point vertex, besides $c$, that is simultaneously adjacent to a pair vertex in $P(x,c)$ and to $b_{(i_1,r_1)}$, then,
    \begin{align}
        \label{lem:normal-component-glb:ineq-3}
        w_{i_1} \leq w(P(i_1,x)) = r_1 + w(P(x,c)).
    \end{align}
    From Lemma~\ref{lem:nice-shortest-path}, we may assume that $P(x,c)$ is nice, thus it can be extended to another nice path $P(x,f_1)$ through the edges $\{c,b_{(i_1,r_1)}\}$ and $\{b_{(i_1,r_1)}, f_1\}$, where $f_1$ is the frontier vertex of $(i_1,r_1)$.
    Then,
    \begin{align}
        \label{lem:normal-component-glb:ineq-4}
        w(P(x, f_1))
         & = w(P(x,c)) + w(c,b_{(i_1,r_1)}) + w(b_{(i_1,r_1)}, f_1)
        = w(P(x,c)) + 2 r_1.
    \end{align}

    \textbf{Case 2.1:} Assume that $r_1 + w(P(x,c)) \leq \frac{7}{2} \overline{r}$.
    Then
    \begin{align*}
        w_{i_1}
         & \stackrel{{\scriptsize \mathrm{(ineq.\ref{lem:normal-component-glb:ineq-3})}}}{\leq} r_1 + w(P(x,c))
        \leq \frac{7}{2} \overline{r}
        \leq \frac{7}{2} \sr(C_d),
    \end{align*}
    and we have proven inequality (\ref{lem:normal-component-glb:ineq-1}).

    \textbf{Case 2.2:} Assume that $r_1 + w(P(x,c)) > \frac{7}{2} \overline{r}$.
    This case is further broken into two subcases.

    \textbf{Case 2.2.1:} Assume that $\overline{r} \geq \frac{1}{5} w(P(x,f_1))$.
    Then
    \begin{align}
        \label{lem:normal-component-glb:ineq-5}
        \min \left\{ \frac{1}{2} w(P(x,f_1)) - \frac{1}{2} \overline{r}, w(P(x,f_1)) - 3 \overline{r} \right\}
         & = w(P(x,f_1)) - 3 \overline{r},
    \end{align}
    which implies that
    \begin{align*}
        \frac{7}{2} \cdot \sr(C_d)
         & \stackrel{{\scriptsize \mathrm{(Lem.\ref{lem:nice-path-collection})}}}{\mathmakebox[\widthof{$\stackrel{{\scriptsize \mathrm{(ineq.\ref{lem:normal-component-glb:ineq-7})}}}{\geq}$}]{\geq}} \frac{7}{8} \left( \rad(G) + \min \left\{ \frac{1}{2} w(P(x,f_1)) - \frac{1}{2} \overline{r}, w(P(x,f_1)) - 3 \overline{r} \right\} \right) \\
         & \stackrel{{\scriptsize \mathrm{(eq.\ref{lem:normal-component-glb:ineq-5})}}}{\mathmakebox[\widthof{$\stackrel{{\scriptsize \mathrm{(ineq.\ref{lem:normal-component-glb:ineq-7})}}}{\geq}$}]{=}} \frac{7}{8} \left( \rad(G) + w(P(x,f_1)) - 3 \overline{r} \right)                                                                        \\
         & \stackrel{{\scriptsize \mathrm{(eq.\ref{lem:normal-component-glb:ineq-4})}}}{\mathmakebox[\widthof{$\stackrel{{\scriptsize \mathrm{(ineq.\ref{lem:normal-component-glb:ineq-7})}}}{\geq}$}]{=}} \frac{7}{8} \left( \rad(G) + w(P(x,c)) + 2 r_1 - 3 \overline{r} \right)                                                                  \\
         & \stackrel{\phantom{\scriptsize \mathrm{(ineq.\ref{lem:normal-component-glb:ineq-7})}}}{\geq} \frac{7}{8} \left( 2 \cdot w(P(x,c)) + 2 r_1 - 3 \overline{r} \right)                                                                                                                                                                       \\
         & \stackrel{\phantom{\scriptsize \mathrm{(ineq.\ref{lem:normal-component-glb:ineq-7})}}}{>} \frac{7}{8} \left( 2 \cdot w(P(x,c)) + 2 r_1 - 3 \cdot \frac{2}{7}(r_1 + w(P(x,c))) \right)                                                                                                                                                    \\
         & \stackrel{\phantom{\scriptsize \mathrm{(ineq.\ref{lem:normal-component-glb:ineq-7})}}}{=} w(P(x,c)) + r_1                                                                                                                                                                                                                                \\
         & \stackrel{{\scriptsize \mathrm{(ineq.\ref{lem:normal-component-glb:ineq-3})}}}{\mathmakebox[\widthof{$\stackrel{{\scriptsize \mathrm{(ineq.\ref{lem:normal-component-glb:ineq-7})}}}{\geq}$}]{\geq}} w_{i_1}.
    \end{align*}
    Thus, we have proven inequality (\ref{lem:normal-component-glb:ineq-1}).

    \textbf{Case 2.2.2:} Assume that $\overline{r} < \frac{1}{5} w(P(x,f_1))$.
    Then
    \begin{align}
        \label{lem:normal-component-glb:ineq-6}
        \min \left\{ \frac{1}{2} w(P(x,f_1)) - \frac{1}{2} \overline{r}, w(P(x,f_1)) - 3 \overline{r} \right\}
         & = \frac{1}{2} w(P(x,f_1)) - \frac{1}{2} \overline{r},
    \end{align}
    and also because $r_1 \leq \overline{r}$ and because of inequality~(\ref{lem:normal-component-glb:ineq-4}), we have that
    \begin{align}
        \label{lem:normal-component-glb:ineq-7}
        3 r_1
         & < w(P(x,c)).
    \end{align}
    Then,
    \begin{align*}
        \frac{7}{2} \cdot \sr(C_d)
         & \stackrel{\phantom{\scriptsize \mathrm{(ineq.\ref{lem:normal-component-glb:ineq-7})}}}{>} \frac{16}{5} \sr(C_d)                                                                                                                                                                                                                          \\
         & \stackrel{{\scriptsize \mathrm{(Lem.\ref{lem:nice-path-collection})}}}{\mathmakebox[\widthof{$\stackrel{{\scriptsize \mathrm{(ineq.\ref{lem:normal-component-glb:ineq-7})}}}{\geq}$}]{\geq}} \frac{4}{5} \left( \rad(G) + \min \left\{ \frac{1}{2} w(P(x,f_1)) - \frac{1}{2} \overline{r}, w(P(x,f_1)) - 3 \overline{r} \right\} \right) \\
         & \stackrel{{\scriptsize \mathrm{(eq.\ref{lem:normal-component-glb:ineq-6})}}}{\mathmakebox[\widthof{$\stackrel{{\scriptsize \mathrm{(ineq.\ref{lem:normal-component-glb:ineq-7})}}}{\geq}$}]{=}} \frac{4}{5} \left( \rad(G) + \frac{1}{2} w(P(x,f_1)) - \frac{1}{2} \overline{r} \right)                                                  \\
         & \stackrel{\phantom{\scriptsize \mathrm{(ineq.\ref{lem:normal-component-glb:ineq-7})}}}{>} \frac{4}{5} \left( \rad(G) + \frac{1}{2} w(P(x,f_1)) - \frac{1}{2} \cdot \frac{1}{5} w(P(x,f_1)) \right)                                                                                                                                       \\
         & \stackrel{{\scriptsize \mathrm{(eq.\ref{lem:normal-component-glb:ineq-4})}}}{\mathmakebox[\widthof{$\stackrel{{\scriptsize \mathrm{(ineq.\ref{lem:normal-component-glb:ineq-7})}}}{\geq}$}]{=}} \frac{4}{5} \left( \rad(G) + \frac{1}{2} w(P(x,c)) + r_1 - \frac{1}{10} (w(P(x,c) + 2 r_1)) \right)                                      \\
         & \stackrel{\phantom{\scriptsize \mathrm{(ineq.\ref{lem:normal-component-glb:ineq-7})}}}{\geq} \frac{4}{5} \left( \frac{3}{2} w(P(x,c)) + r_1 - \frac{1}{10} (w(P(x,c) + 2 r_1)) \right)                                                                                                                                                   \\
         & \stackrel{\phantom{\scriptsize \mathrm{(ineq.\ref{lem:normal-component-glb:ineq-7})}}}{=} \frac{28}{25} w(P(x,c)) + \frac{16}{25} r_1                                                                                                                                                                                                    \\
         & \stackrel{{\scriptsize \mathrm{(ineq.\ref{lem:normal-component-glb:ineq-7})}}}{\mathmakebox[\widthof{$\stackrel{{\scriptsize \mathrm{(ineq.\ref{lem:normal-component-glb:ineq-7})}}}{\geq}$}]{>}} w(P(x,c)) + r_1                                                                                                                        \\
         & \stackrel{{\scriptsize \mathrm{(ineq.\ref{lem:normal-component-glb:ineq-3})}}}{\mathmakebox[\widthof{$\stackrel{{\scriptsize \mathrm{(ineq.\ref{lem:normal-component-glb:ineq-7})}}}{\geq}$}]{\geq}} w_{i_1},                                                                                                                            \\
    \end{align*}
    and we have proven inequality (\ref{lem:normal-component-glb:ineq-1}).

    Since we always have that $w_{i_1} \leq \frac{7}{2}$ in Case~2, we have proven that inequality~(\ref{lem:normal-component-glb:ineq-1}) holds, and thus completing the proof.
\end{proof}

\section{Proof of Lemma~\ref{lem:outliers-final-cost-bound}}

To prove Lemma~\ref{lem:outliers-final-cost-bound}, we need results that are similar to Lemma~\ref{lem:normal-component}, which works for the version without outliers.
The following Lemmas~\ref{lem:outliers-dual-cost-bound} and~\ref{lem:outliers-parital-cost-bound} serve this purpose.

\begin{lem}
    \label{lem:outliers-dual-cost-bound}
    Let $\calb' \subseteq \calb$ be a set of pairwise disjoint tight pairs and let $U \subseteq X'$ be a set of points not covered by $\calb'$ such that every $j \in U$ is tight.
    Then
    \begin{align*}
        \sr(\calb')
         & \leq \sum_{j \in X'} \alpha_j - |\calb'| \cdot \lambda - |U| \cdot \gamma.
    \end{align*}
\end{lem}
\begin{proof}
    Since pairs in $\calb'$ are tight, it follows that $r = \sum_{j \in B(i,r)} \alpha_j - \lambda$ for each pair $(i,r) \in \calb'$, and the pairs are pairwise disjoint, we have that no point is charged more than once this way, and since only points contained in these balls are charged, it follows that a value of $\gamma$ can be charged from each point in $U$.
    Hence,
    \begin{align*}
        \sr(\calb')
         & = \sum_{(i,r) \in \calb'} r + 0
        = \sum_{(i,r) \in \calb'} \left( \sum_{j \in B(i,r)} \alpha_j - \lambda \right) + \sum_{j \in U} \alpha_j - |U| \cdot \gamma
        \leq \sum_{j \in X'} \alpha_j - |\calb'| \cdot \lambda - |U| \cdot \gamma,
    \end{align*}
    completing the proof.
\end{proof}

Let $\calb' \subset \calb$ be a set of pairs, then we denote by $\creplaced{\calb'}$ the set of pairs obtained by replacing each component $C \in \comp(\calb')$ for a cheapest pair $(i,r)$ such that $X(C) \subseteq B(i,r)$.
Also, let $\OPT'_{LP}$ denote the value of an optimal fractional solution to $(D')$.

\begin{lem}
    \label{lem:outliers-parital-cost-bound}
    Let $\calb_1, \calb_2 \subseteq \calb$ be sets of tight pairs.
    If $|\comp(\calb_1)| \geq k'$ and $|\out(\calb_1)| \geq m$, and every $j \in \out(\calb_1)$ is tight, then
    \begin{align*}
        \sr(\creplaced{\calb_1 \cup \calb_2})
         & \leq 3 \cdot \OPT'_{LP} + 3 \cdot \sr(\calb_2)
    \end{align*}
\end{lem}
\begin{proof}
    For each component $C$, let $(i_C,r_C)$ denote a pair with smallest radius $r_C$ such that $X(C) \subseteq B(i_C,r_C)$, and let $C_d$ denote a subset of $C$ whose corresponding balls are pairwise disjoint and maximizes~$\sr(C_d)$.

    Define $F = \bigcup_{C \in \comp(\calb_1)} C_d$, $F' = \bigcup_{C \in \comp(\calb_2)} C_d$, and $F'' = \bigcup_{C \in \comp(\calb_1 \cup \calb_2)} C_d$, and observe that the balls corresponding to the pairs in $F$ are pairwise disjoint, and, similarly, the balls corresponding to the pairs in $F'$ and the pairs in $F''$ are also respectively pairwise disjoint.
    Also, notice that $\sr(F'') \leq \sr(F) + \sr(F') \leq \sr(F) + \sr(\calb_2)$.
    This implies that $|F| \geq |\comp(\calb_1)| \geq k'$, as $F$ contains at least a pair for each component in $\comp(\calb_1)$, and it also implies that $\out(F) \supseteq \out(\calb_1)$ because $\calb_1$ covers every point $F$ covers, hence we may apply Lemma~\ref{lem:outliers-dual-cost-bound}.
    Hence,
    \begin{align*}
        \sr(\creplaced{\calb_1 \cup \calb_2})
         & \stackrel{\phantom{\scriptsize \mathrm{(Lem.\ref{lem:outliers-dual-cost-bound})}}}{=} \sum_{C \in \comp(\calb_1 \cup \calb_2)} r_C
        \stackrel{{\scriptsize \mathrm{(Lem.\ref{lem:normal-component})}}}{\leq} \sum_{C \in \comp(\calb_1 \cup \calb_2)} 3 \cdot \sr(C_d)
        = 3 \cdot \sr(F'')                                                                                                                                                                                                   \\
         & \stackrel{\phantom{\scriptsize \mathrm{(Lem.\ref{lem:outliers-dual-cost-bound})}}}{\leq} 3 \cdot \sr(F) + 3 \cdot \sr(\calb_2)                                                                                    \\
         & \stackrel{{\scriptsize \mathrm{(Lem.\ref{lem:outliers-dual-cost-bound})}}}{\leq} 3 \left( \sum_{j \in X'} \alpha_j - \left| F \right| \cdot \lambda - |\out(\calb_1)| \cdot \gamma \right) + 3 \cdot \sr(\calb_2) \\
         & \stackrel{\phantom{\scriptsize \mathrm{(Lem.\ref{lem:outliers-dual-cost-bound})}}}{\leq} 3 \left( \sum_{j \in X'} \alpha_j - k \cdot \lambda - m \cdot \gamma \right) + 3 \cdot \sr(\calb_2)                      \\
         & \stackrel{\phantom{\scriptsize \mathrm{(Lem.\ref{lem:outliers-dual-cost-bound})}}}{\leq} 3 \cdot \OPT'_{LP} + 3 \cdot \sr(\calb_2),
    \end{align*}
    and this completes the proof.
\end{proof}

Recall that for an ordered set of pairs $\calb' = \{(i_1,r_1),\dots\}$, define $\calb'_q := \{(i_{q'}, r_{q'}) \in \calb' : q' \leq q \}$.

\defnoutliersorderlystructured*

Because the removal of a tight pair from $\calb'_\ell$ can decrease the number of components by at most~one, we write for reference the following extra property that is a direct consequence of~\ref{defn:os4} of orderly structured sets:
\begin{enumerate}[label=OS$\arabic*.$,ref={\mbox{\rm OS$\arabic*$}}, noitemsep, leftmargin=1.5cm]
    \setcounter{enumi}{3}
    \item \label{defn:os5} $|\comp(\calb'_{\ell-1})| \geq k$.
\end{enumerate}

\lemoutliersfinalcostbound*
\begin{proof}
    The proof follows by applying Lemma~\ref{lem:outliers-parital-cost-bound} in four distinct cases.
    For all cases, we note that from~\ref{defn:os2} we always have that $\alpha_j = \gamma$ for every $j \in \out(\calb'_i)$ and every $i \geq \ell - 1$.

    \textbf{Case 1:} Assume that $\ell' = \ell$.
    Then, we have that
    \begin{align*}
        |\comp(\calb'_{\ell-1})| \stackrel{({\scriptsize\ref{defn:os5}})}{\geq} & {\mathmakebox[\widthof{$m$}]{k}} \stackrel{({\scriptsize\ref{defn:os4}})}{\geq} |\comp(\calb'_{\ell-1} \cup \{(i_\ell,r_\ell),(i^*,r^*)\})|, \text{ and} \\
        |\out(\calb'_{\ell-1})| \stackrel{({\scriptsize\ref{defn:os3}})}{>}     & m \stackrel{({\scriptsize\ref{defn:os3}})}{\geq} |\out(\calb'_{\ell-1} \cup \{(i_\ell,r_\ell),(i^*,r^*)\})|.
    \end{align*}
    Thus, $\creplaced{\calb'_{\ell} \cup \{(i^*,r^*)\}}$ is a feasible solution and its cost can be bounded:
    \begin{align*}
        \sr(\creplaced{\calb'_{\ell} \cup \{(i^*,r^*)\}})
         & = \sr(\creplaced{\calb'_{\ell-1} \cup \{(i_\ell,r_\ell),(i^*,r^*)\}})
        \stackrel{{\scriptsize \mathrm{(Lem.\ref{lem:outliers-parital-cost-bound})}}}{\leq} 3 \cdot \OPT'_{LP} + 3 (r_\ell + r^*) \\
         & \leq 3 \cdot \OPT'_{LP} + O(\epsilon) \cdot \OPT.
    \end{align*}

    \textbf{Case 2:} Assume that $\ell' = \ell - 1$.
    Then, we have that
    \begin{align*}
        |\comp(\calb'_{\ell-1})| \stackrel{({\scriptsize\ref{defn:os5}})}{\geq} & {\mathmakebox[\widthof{$m$}]{k}} \stackrel{({\scriptsize\ref{defn:os4}})}{\geq} |\comp(\calb'_{\ell-1} \cup \{(i^*,r^*)\})|, \text{ and} \\
        |\out(\calb'_{\ell-1})| \stackrel{({\scriptsize\ref{defn:os3}})}{>}     & m \stackrel{({\scriptsize\ref{defn:os3}})}{\geq} |\out(\calb'_{\ell-1} \cup \{(i^*,r^*)\})|.
    \end{align*}
    Thus, $\calb'_{\ell-1}\cup \{(i^*,r^*)\}$ is a feasible solution and its cost can be bounded:
    \begin{align*}
        \sr(\creplaced{\calb'_{\ell-1} \cup \{(i^*,r^*)\}})
         & \stackrel{{\scriptsize \mathrm{(Lem.\ref{lem:outliers-parital-cost-bound})}}}{\leq} 3 \cdot \OPT'_{LP} + 3 \cdot r^*
        \leq 3 \cdot \OPT'_{LP} + O(\epsilon) \cdot \OPT.
    \end{align*}

    \textbf{Case 3:} Assume that $\ell' \leq \ell-2$ and that $|\comp(\calb'_{\ell} \cup \{(i^*,r^*)\})| > k$.
    Because $|\comp(\calb'_{\ell} \cup \{(i^*,r^*)\})| > k$, from~\ref{defn:os4} we have that there exists $h \in [\ell', \ell-1]$, such that $|\comp(\calb'_{h+1} \cup \{(i^*,r^*)\})| > {\mathmakebox[\widthof{$m$}]{k}} \geq |\comp(\calb'_{h} \cup \{(i^*,r^*)\})|$.
    Observe that $(\calb'_{h+1} \cup \{(i^*,r^*)\}) \setminus (\calb'_{h} \cup \{(i^*,r^*)\}) = \{ (i_{h+1},r_{h+1}) \}$, and since the inclusion of $(i_{h+1},r_{h+1})$ increases the number of components, we have that ${B(i_{h+1},r_{h+1})}$ does not intersect the ball corresponding to any pair in $\calb'_{h} \cup \{(i^*,r^*)\}$.
    Therefore, the number of components increases by at most one and we have that $k = |\comp(\calb'_{h} \cup \{(i^*,r^*)\})|$.
    We further break this case into two subcases.

    \textbf{Case 3.1:} Assume that the $B(i^*,r^*)$ intersects the ball corresponding to some pair in $\calb'_h$.
    Then, the removal of $\{(i^*,r^*)\}$ from $\calb'_{h} \cup \{(i^*,r^*)\}$ cannot decrease the number of components and we have that $|\comp(\calb'_{h})| \geq |\comp(\calb'_{h} \cup \{(i^*,r^*)\})| = k$.
    Then, we have that
    \begin{align*}
        |\comp(\calb'_{h})| \stackrel{\phantom{({\scriptsize\ref{defn:os3}})}}{\geq} & {\mathmakebox[\widthof{$m$}]{k}} \stackrel{\phantom{({\scriptsize\ref{defn:os3}})}}{=} |\comp(\calb'_{h} \cup \{(i^*,r^*)\})|, \text{ and} \\
        |\out(\calb'_{h})| \stackrel{({\scriptsize\ref{defn:os3}})}{>}               & m \stackrel{({\scriptsize\ref{defn:os3}})}{\geq} |\out(\calb'_{h} \cup \{(i^*,r^*)\})|.
    \end{align*}
    Thus, $\calb'_{h} \cup \{(i^*,r^*)\}$ is a feasible solution and its cost can be bounded:
    \begin{align*}
        \sr(\creplaced{\calb'_{h} \cup \{(i^*,r^*)\}})
         & \stackrel{{\scriptsize \mathrm{(Lem.\ref{lem:outliers-parital-cost-bound})}}}{\leq} 3 \cdot \OPT'_{LP} + 3 \cdot r^*
        \leq 3 \cdot \OPT'_{LP} + O(\epsilon) \cdot \OPT.
    \end{align*}

    \textbf{Case 3.2:} Assume that the $B(i^*,r^*)$ does not intersect the ball corresponding to any pair in $\calb'_h$.
    Then we can replace $(i^*,r^*)$ by $(i_{h+1},r_{h+1})$ in $\calb'_h \cup \{(i^*,r^*)\}$ without increasing the number of components, because $B(i_{h+1},r_{h+1})$ also does not intersect $B(i,r)$ for any pair $(i,r) \in \calb'_h$, i.e., $k = |\comp(\calb'_h \cup \{(i^*,r^*)\})| = |\comp(\calb'_h \cup \{(i_{h+1},r_{h+1})\})|$.
    Then, we have that
    \begin{align*}
        |\comp(\calb'_h \cup \{(i_{h+1},r_{h+1})\})| \stackrel{\phantom{({\scriptsize\ref{defn:os3}})}}{=} & {\mathmakebox[\widthof{$m$}]{k}} \stackrel{\phantom{({\scriptsize\ref{defn:os3}})}}{=} |\comp(\calb'_{h} \cup \{(i^*,r^*)\})|, \text{ and} \\
        |\out(\calb'_h \cup \{(i_{h+1},r_{h+1})\})| \stackrel{({\scriptsize\ref{defn:os3}})}{>}            & m \stackrel{({\scriptsize\ref{defn:os3}})}{\geq} |\out(\calb'_{h} \cup \{(i^*,r^*)\})|.
    \end{align*}
    Furthermore, since $B(i^*,r^*)$ and $B(i_{h+1},r_{h+1})$ do not intersect, we have that
    \begin{align}
        \label{lem:outliers-final-cost-bound:ineq1}
        \sr(\creplaced{\calb'_h \cup \{(i^*,r^*)\}}) \leq \sr(\creplaced{\calb'_h \cup \{(i^*,r^*)\}}) + \sr(\{(i_{h+1,r_{h+1}})\}) = \sr(\creplaced{\calb'_h \cup \{(i_{h+1},r_{h+1}),(i^*,r^*)\}})
    \end{align}
    Thus, $\creplaced{\calb'_{h} \cup \{(i^*,r^*)\}}$ is a feasible solution and its cost can be bounded:
    \begin{align*}
        \sr(\creplaced{\calb'_{h} \cup \{(i^*,r^*)\}})
         & \stackrel{{\scriptsize \mathrm{(ineq.\ref{lem:outliers-final-cost-bound:ineq1})}}}{\mathmakebox[\widthof{$\stackrel{{\scriptsize \mathrm{(ineq.\ref{lem:outliers-final-cost-bound:ineq1})}}}{\leq}$}]{\leq}} \sr(\creplaced{\calb'_h \cup \{(i_{h+1},r_{h+1}),(i^*,r^*)\}}) \\
         & \stackrel{\phantom{\scriptsize \mathrm{(ineq.\ref{lem:outliers-final-cost-bound:ineq1})}}}{=} \sr(\creplaced{\calb'_{h+1} \cup \{(i^*,r^*)\}})                                                                                                                              \\
         & \stackrel{{\scriptsize \mathrm{(Lem.\ref{lem:outliers-parital-cost-bound})}}}{\mathmakebox[\widthof{$\stackrel{{\scriptsize \mathrm{(ineq.\ref{lem:outliers-final-cost-bound:ineq1})}}}{\leq}$}]{\leq}} 3 \cdot \OPT'_{LP} + 3 \cdot r^*                                    \\
         & \stackrel{\phantom{\scriptsize \mathrm{(ineq.\ref{lem:outliers-final-cost-bound:ineq1})}}}{\leq} 3 \cdot \OPT'_{LP} + O(\epsilon) \cdot \OPT.
    \end{align*}

    \textbf{Case 4:} Assume $\ell' \leq \ell-2$ and that $|\comp(\calb'_{\ell} \cup \{(i^*,r^*)\})| \leq k$.
    Then, we have that
    \begin{align*}
        |\comp(\calb'_{\ell-1})| \stackrel{({\scriptsize\ref{defn:os5}})}{\geq} & {\mathmakebox[\widthof{$m$}]{k}} \stackrel{({\scriptsize\ref{defn:os4}})}{\geq} |\comp(\calb'_{\ell-1} \cup \{(i_\ell,r_\ell),(i^*,r^*)\})|, \text{ and} \\
        |\out(\calb'_{\ell-1})| \stackrel{({\scriptsize\ref{defn:os3}})}{>}     & m \stackrel{({\scriptsize\ref{defn:os3}})}{\geq} |\out(\calb'_{\ell-1} \cup \{(i_\ell,r_\ell),(i^*,r^*)\})|.
    \end{align*}
    Thus $\creplaced{\calb'_{\ell-1}\cup \{(i_{\ell},r_\ell),(i^*,r^*)\}}$ is a feasible solution and its cost can be bounded:
    \begin{align*}
        \sr(\creplaced{\calb'_{\ell-1} \cup \{(i^*,r^*)\}})
         & \stackrel{{\scriptsize \mathrm{(Lem.\ref{lem:outliers-parital-cost-bound})}}}{\leq} 3 \cdot \OPT'_{LP} + 3 \cdot r^*
        \leq 3 \cdot \OPT'_{LP} + O(\epsilon) \cdot \OPT.
    \end{align*}
    This completes the proof.
\end{proof}

\section{Proof of Lemma~\ref{lem:outliers-find-os-polytime}}

Let $s^{*}$ denote the iteration in which we are given an interval $(b_{L}^{(s^{*})},b_{R}^{(s^{*})})$ such that the execution of the subroutine for any $\lambda \in (b_{L}^{(s^{*})},b_{R}^{(s^{*})})$ ends at phase $s^*-1$ and $g_{s^*-1}(\lambda)$ is affine.
Let $\calp_M$ be the computed pair when running the subroutine for $\lambda \in (b_{L}^{(s^{*})},b_{R}^{(s^{*})})$, and let $\calp_L$ and $\calp_R$ be the computed pairs for $\lambda = b_{L}^{(s^{*})}$ and $\lambda = b_{R}^{(s^{*})}$.
Ahmadian and Swamy showed that our subroutine has the following continuity property~\cite{Ahmadian2016}, which is an adaptation of the continuity property presented by Charikar and Panigrahy~\cite{Charikar2004}.
We provide a proof for completeness.

\begin{lem}[Ahmadian and Swamy~\cite{Ahmadian2016} - rephrased]
    \label{lem:outliers-continuity}
    Let $(\lambda,\gamma,\alpha)$ and $(\lambda',\gamma',\alpha')$ be solutions to $(D')$ computed by our subroutine for $\lambda$ and $\lambda'$.
    If $|\lambda' - \lambda| \leq \delta$, then
    \begin{enumerate}[label=$\arabic*.$,ref={\mbox{\rm $\arabic*$}}, noitemsep, leftmargin=1.5cm]
        \item \label{lem:outliers-continuity-1} $|\alpha'_j - \alpha_j| \leq \delta \cdot 2^{|X'|}$ for every $j \in X'$,
        \item \label{lem:outliers-continuity-2} $|\gamma' - \gamma| \leq \delta \cdot 2^{|X'|}$,
        \item \label{lem:outliers-continuity-3} if a pair $(i,r)$ is tight for $(\lambda,\gamma,\alpha)$, then $\sum_{j \in B(i,r) \cap X'} \alpha'_j \geq r + \lambda' - \delta \cdot 2^{|X'|}$, and
        \item \label{lem:outliers-continuity-4} if a pair $(i',r')$ is tight for $(\lambda',\gamma',\alpha')$, then $\sum_{j \in B(i',r') \cap X'} \alpha_j \geq r' + \lambda - \delta \cdot 2^{|X'|}$.
    \end{enumerate}
\end{lem}
\begin{proof}
    For each $j \in X'$, let $\alpha^*_j = \min\{\alpha_j, \alpha'_j\}$, and, to simplify notation, denote each point by a number in $\{1,2,\dots,|X'|\}$ in such a way that $\alpha^*_1 \leq \alpha^*_2 \leq \dots \leq \alpha^*_{|X'|}$.
    Consider the two executions of the subroutine in parallel, one with $\lambda$ and the other with $\lambda'$, and recall that the variables $\alpha_j$ and $\alpha'_j$ start at zero for every $j \in X'$ and each one increases at every phase of the routine until a pair covering it becomes tight.

    We apply induction that $|\alpha'_j - \alpha_j| \leq \delta \cdot 2^{j-1}$.
    Assume, w.l.o.g., that $\alpha^*_1 = \alpha_1$, and let $t_1$ denote the time at which a pair covering point $1$ becomes tight for the execution with $\lambda$.
    Since this is the first phase, it is also the time at which the first phase ends.
    Since this is the end of the first phase, every variable was increasing up to this time.
    This means that at time $t_1$ we have that
    \begin{align*}
        \sum_{j \in B(i,r) \cap X'} \alpha'_j
         & = \sum_{j \in B(i,r) \cap X'} \alpha_j
        = r + \lambda
        \leq r + \lambda' + \delta,
    \end{align*}
    for every pair $(i,r)$ covering point $1$, and hence $\alpha'_1$ can increase by at most $\delta$ until the routine ends, and we conclude that $|\alpha'_1 - \alpha_1| \leq \delta \cdot 2^{0}$

    Assume that the inductive hypothesis is true for every $j \in [h - 1]$.
    Let $t_h$ be the lowest time a pair covering point $h$ becomes tight in one of the two executions, and assume that this occurs for execution with $\lambda$.
    The calculations are the same if this had occurred for the execution with $\lambda'$.
    From the inductive hypothesis we have that $|\alpha'_j \leq \alpha_j| + \delta \cdot 2^{j-1}$ for $j < h$, and since the variable $\alpha_j$ and $\alpha'_j$ for every point $j \geq h$ never stopped increasing before time $t_h$ we have that $\alpha'_j = \alpha_j$ for $j \geq h$.
    This means that time $t_h$ we have that
    \begin{align*}
        \sum_{j \in B(i,r) \cap X'} \alpha'_j
         & = \sum_{\substack{j \in B(i,r) \cap X'                      \\ j < h}} \alpha'_j + \sum_{\substack{j \in B(i,r) \cap X' \\ j \geq h}} \alpha'_j
        \leq \sum_{\substack{j \in B(i,r) \cap X'                      \\ j < h}} \left( \alpha_j + \delta \cdot 2^{j-1} \right) + \sum_{\substack{j \in B(i,r) \cap X' \\ j \geq h}} \alpha'_j \\
         & = r + \lambda + \sum_{\substack{j \in B(i,r) \cap X'        \\ j < h}} \delta \cdot 2^{j-1}
        \leq r + \lambda' + \delta + \sum_{j < h} \delta \cdot 2^{j-1} \\
         & = r + \lambda' + \delta \cdot 2^{h-1},
    \end{align*}
    for every pair $(i,r)$ containing point $h$, and hence $\alpha'_h$ can increase by at most $\delta \cdot 2^{h-1}$ until the routine ends, and we conclude that $|\alpha'_h - \alpha_h| \leq \delta \cdot 2^{h-1}$.
    Now, recall that the routine ends when the tight pairs in $\calp$ do not cover at most $m$ points in $X'$.
    This completes the induction and proves property~\ref{lem:outliers-continuity-1}.

    Now, let $(i,r)$ be a pair that is tight for $(\lambda, \gamma, \alpha)$, then
    \begin{align*}
        r + \lambda'
         & \leq r + \lambda + \delta
        = \sum_{j \in B(i,r) \cap X'} \alpha_j + \delta
        \leq \sum_{j \in B(i,r) \cap X'} (\alpha'_j + \delta \cdot 2^{j-1}) + \delta                    \\
         & \leq \sum_{j \in B(i,r) \cap X'} \alpha'_j + \sum_{i=1}^{|X'|} \delta \cdot 2^{j-1} + \delta
        \leq \sum_{j \in B(i,r) \cap X'} \alpha'_j + \delta \left( 2^{|X'|} - 1 \right) + \delta        \\
         & \leq \sum_{j \in B(i,r) \cap X'} \alpha'_j + \delta \cdot 2^{|X'|},
    \end{align*}
    and we have proven property~\ref{lem:outliers-continuity-3}
    The analogous calculations for a pair $(i',r')$ be a pair that is tight for $(\lambda', \gamma', \alpha')$, one can prove~property~\ref{lem:outliers-continuity-4}.

    Recall that the routine ends when the tight pairs in $\calp$ do not cover at most $m$ points in $X'$.
    Suppose, w.l.o.g., that at time $t$ the routine with $\lambda$ ends.
    Since there was at least one $j \in X'$ such that $\alpha_j$ was increasing up to time $t$, this execution sets $\gamma := \max_{j \in X'} \alpha_j = t$.
    Now, since $|\alpha'_j - \alpha_j| \leq \delta \cdot 2^{|X'|}$ for every $j \in X'$, we then have that
    \begin{align*}
        \gamma'
         & = \max_{j \in X'} \alpha'_j
        \leq \max_{j \in X'} \alpha_j + \delta \cdot 2^{|X'|}
        = \gamma + \delta \cdot 2^{|X'|},
    \end{align*}
    and we have proven property~\ref{lem:outliers-continuity-2}.
\end{proof}

Let $\calb'_M$, $\calb'_L$, and $\calb'_R$ be respectively the sets of all tight pairs computed by running the subroutine for $\lambda \in (b_{L}^{(s^{*})},b_{R}^{(s^{*})})$, $\lambda = b_{L}^{(s^{*})}$, and $\lambda = b_{R}^{(s^{*})}$.
Because we can choose $\lambda \in (b_{L}^{(s^{*})},b_{R}^{(s^{*})})$ to be arbitrarily close to $b_{L}^{(s^{*})}$ or to $b_{R}^{(s^{*})}$, it follows from the continuity property from Lemma~\ref{lem:outliers-continuity} that  $\calp_M \subseteq \calb'_M \subseteq \calb'_L$, and $\calp_M \subseteq \calb'_M \subseteq \calb'_R$.
Now we are ready to prove Lemma~\ref{lem:outliers-find-os-polytime}.

\lemoutliersfindospolytime*
\begin{proof}
    Assume that $|\comp(\calp_M)| \leq k$ and consider $\calp_L$, for which we have that $|\comp(\calp_L)| > k$.
    The case when $|\comp(\calp_M)| > k$ is analogous but we consider $\calp_R$ instead of $\calp_L$, as $|\comp(\calp_R)| \leq k$.

    Let $p_s \in \calp_L$ be the last pair included in $\calp_L$ and $p'_{s'} \in \calp_M$ be the last pair included in $\calp_M$ during the execution of their respective subroutines (it could be that $p_s = p'_{s'}$).
    Let $\calp' = (\calp_L \setminus \{p_s\}) \cap (\calp_M \setminus \{p'_{s'}\})$, and let $\calp_L \setminus \{p_s\} \setminus \calp_M = \{p_1,\dots,p_{s-1}\}$ and $\calp_M \setminus \calp_L = \{p'_1,\dots,p'_{s'-1}\}$ be the pairs exclusively in $\calp_L$ and in $\calp_M$ except for their respective last pairs.

    Consider the following covering procedure where we are given a sequence of the pairs in $\calp_L \cup \calp_M$.
    Starting with $\calq := \calp'$, iteratively include in $\calq$, one at a time, the pairs in the given sequence according to its ordering until at most $m$ points are not covered by the tight pairs in $\calq$.
    Then we return $\calq$.

    Because $|\out(\calp_L \setminus \{p_s\})| > m \geq |\out(\calp_L)|$, if the sequence given is $(p_1,\dots,p_{s-1},p_s)$ then the procedure ends exactly when $p_s$ is included in $\calq$, hence $\calq = \calp_L$.
    Note that, in particular, the covering procedure yields $\calq = \calp_L$ for any sequence whose first $h$ pairs correspond exactly to the sequence $(p_1,\dots,p_{s-1},p_s)$.
    Similarly, because $|\out(\calp_M \setminus \{p'_{s'}\})| > m \geq |\out(\calp_M)|$, if the sequence given is $(p'_1,\dots,p'_{s'-1},p'_{s'})$ then the procedure ends exactly when $p'_{s'}$ is included in $\calq$, hence $\calq = \calp_M$.
    Again, note that, in particular, the covering procedure yields $\calq = \calp_M$ for any sequence whose first $s'$ pairs correspond exactly to the sequence $(p'_1,\dots,p'_{s'-1},p'_{s'})$.

    Define $S^{(i)} := (p'_{i},p'_{i+1},\dots,p'_{s'},p_1,\dots,p_{s-1},p_s)$ for $i \in [s']$, and $S^{(s'+1)} := (p_1,\dots,p_{s-1},p_s)$, and let $\calq^{(i)}$ denote the set of pairs that is returned by the covering procedure when it is given the sequence $S^{(i)}$ and let its elements be ordered in the order they were included in it.
    Notice that $\calq^{(s'+1)} = \calp_L$ and that $\calq^{1} = \calq_M$ because the first $s'$ pairs is exactly the sequence $(p'_1,\dots,p'_{s'-1},p'_{s'})$.
    Since $|\comp(\calq^{(s'+1)})| > k \geq |\comp(\calq^{(1)})|$, there must be an $i^*$ for which for which $|\comp(\calq^{(i^*+1)})| > k \geq |\comp(\calq^{(i^*)})|$.

    Let $\calb' := \calq^{(i^*+1)}$, then add $p'_{i^*}$ into $\calb'$ as its new last pair in the ordering of $\calb'$.
    Since $\calb' = \{x_1,x_2,\dots\}$ is ordered, recall that $\calb'_{q} := \{x_{q'} \in \calb' : q' \leq q\}$ denotes ordered set of the first $q$ pairs in $\calb'$.
    Let $p$ be the last pair in $\calq^{i^*}$, and define $\ell' := 0$ if $p = p'_{i^*}$, otherwise define $\ell'$ to be the index of $p$ in $\calb'$, and define $\ell$ to be the index of the last pair in $\calb'$.
    Then, $1 \leq \ell$ and $0 \leq \ell' \leq \ell$, and
    \[
        |\comp(\calb'_{\ell})| > k \geq |\comp(\calb'_{\ell'} \cup \{p'_{i^*}\})|.
    \]
    Since the covering procedure stops exactly the moment the number of points not covered by the included pairs is at most $m$, we also have that
    \[
        |\out(\calb'_h)| > m \text{ for every } h < \ell, \text{ and } m \geq |\out(\calb'_{\ell'} \cup \{p'_{i^*}\})|,
    \]
    and since the inclusion of additional pairs in $\calb'_{\ell'} \cup \{p'_{i^*}\}$ can only reduce the number of points that are not covered, we have that
    \[
        m \geq |\out(\calb'_{h'} \cup \{p'_{i^*}\})| \text{ for } \ell' \leq h' \leq \ell.
    \]
    Finally, the continuity property from Lemma~\ref{lem:outliers-continuity} implies that every pair in $\calb' \subseteq \calp_L \cup \calp_M$ is tight for the dual solution $(\lambda, \gamma, \alpha)$ computed by the subroutine when $\lambda = \calb_L^{(s^*)}$.
    This is because we can choose $\lambda \in (b_{L}^{(s^{*})},b_{R}^{(s^{*})})$ to be as arbitrarely close to $b_{L}^{(s^{*})}$ and still obtain $\calp_M$.
    From the definition of the subroutine, every point not covered by $\calp_L$ is tight when it is executed for $\lambda = b_{L}^{(s^{*})}$, and every point not covered by $\calp_M$ is tight when it is executed for $\lambda \in (b_{L}^{(s^{*})},b_{R}^{(s^{*})})$.
    Thus, again, by the continuity property from Lemma~\ref{lem:outliers-continuity}, it follows that every point not covered by $\calp_L$ or $\calp_M$ is also tight for $(\lambda, \gamma, \alpha)$.
    Therefore,
    \[
        \text{every point } j \in \out(\calb'_{\ell'}) \text{ is tight and every pair } (i,r) \in \calb' \text{ is tight},
    \]
    and we conclude that $\calb'$ is an orderly structured set for the dual solution $(\lambda, \gamma, \alpha)$ and that it can be computed in polynomial time.
\end{proof}

\end{document}